\pdfoutput=1

\documentclass[lettersize,journal]{IEEEtran}
\usepackage{amsmath,amsfonts}
\usepackage{algorithmic}
\usepackage{algorithm}
\usepackage{array}
\usepackage{textcomp}
\usepackage{stfloats}
\usepackage{url}
\usepackage{verbatim}
\usepackage{graphicx}
\usepackage{cite}

\usepackage{color}
\usepackage{subcaption}
\newtheorem{theorem}{Theorem}
\newtheorem{lemma}{Lemma}
\newtheorem{corollary}{Corollary}[section]
\newenvironment{proof}{{\noindent\it Proof:}\quad}{\hfill $\square$\par}

\begin{document}

\title{Towards Timely Video Analytics Services at the Network Edge}

\author{Xishuo Li,~\IEEEmembership{Student Member,~IEEE,}
        Shan Zhang,~\IEEEmembership{Member,~IEEE,}
        Yuejiao Huang,~\IEEEmembership{Student Member,~IEEE,}
        \\
        Xiao Ma,~\IEEEmembership{Member,~IEEE,}
        Zhiyuan Wang,~\IEEEmembership{Member,~IEEE,}
        Hongbin Luo,~\IEEEmembership{Member,~IEEE}
\thanks{Part of the results in this paper appeared in IEEE ICWS 2023 \cite{DBLP:conf/icws/XSLi23}.}

\thanks{This work was supported in part by the National Key R\&D Program of China under Grant 2022YFB4501000, in part by the Nature Science Foundation of China under Grant 62271019, 62225201, 62202021, 62372061, in part by the Fundamental Research Funds for the Central Universities, China, and State Key Laboratory of Complex \& Critical Software Environment. \textit{(Corresponding author: Shan Zhang)}}
\thanks{Xishuo Li, Shan Zhang, Yuejiao Huang, Zhiyuan Wang and Hongbin Luo are with State Key Laboratory of Complex \& Critical Software Environment, Computer Science and Engineering, Beihang University, Beijing, 10086, China. (e-mail: lixishuo@buaa.edu.cn; zhangshan18@buaa.edu.cn; SY2106125@buaa.edu.cn; zhiyuanwang@buaa.edu.cn; luohb@buaa.edu.cn)}
\thanks{Xiao Ma is with the State Key Laboratory of Networking and Switching Technology, Beijing University of Posts and Telecommunications, Beijing 100876, China. (e-mail: maxiao18@bupt.edu.cn)}
\thanks{Shan Zhang, Yuejiao Huang, Zhiyuan Wang and Hongbin Luo are also with Zhongguancun Laboratory, Beijing, China.}
}

\markboth{Journal of \LaTeX\ Class Files,~Vol.~14, No.~8, August~2021}%
{Shell \MakeLowercase{\textit{et al.}}: A Sample Article Using IEEEtran.cls for IEEE Journals}


\maketitle

\begin{abstract}
Real-time video analytics services aim to provide users with accurate recognition results timely. However, existing studies usually fall into the dilemma between reducing delay and improving accuracy. The edge computing scenario imposes strict transmission and computation resource constraints, making balancing these conflicting metrics under dynamic network conditions difficult. In this regard, we introduce the age of processed information (AoPI) concept, which quantifies the time elapsed since the generation of the latest accurately recognized frame. AoPI depicts the integrated impact of recognition accuracy, transmission, and computation efficiency. We derive closed-form expressions for AoPI under preemptive and non-preemptive computation scheduling policies w.r.t. the transmission/computation rate and recognition accuracy of video frames. We then investigate the joint problem of edge server selection, video configuration adaptation, and bandwidth/computation resource allocation to minimize the long-term average AoPI over all cameras. We propose an online method, i.e., Lyapunov-based block coordinate descent (LBCD), to solve the problem, which decouples the original problem into two subproblems to optimize the video configuration/resource allocation and edge server selection strategy separately. We prove that LBCD achieves asymptotically optimal performance. According to the testbed experiments and simulation results, LBCD reduces the average AoPI by up to 10.94X compared to state-of-the-art baselines.
\end{abstract}


\section{Introduction}
\IEEEPARstart{V}{ideo} analytics is a key component in many modern applications, including public security, traffic monitoring, patient healthcare, and industrial automation \cite{DBLP:conf/icws/XSLi23,DBLP:journals/corr/abs-2303-14329}. As a transmission- and computation-intensive service, video analytics can greatly benefit from mobile edge computing (MEC), where videos from different cameras are uploaded to nearby edge servers for recognition, as shown in Fig. \ref{Fig.system}. To better utilize the limited transmission and computation capacity of edge networks, numerous studies have been conducted to flexibly adjust video configurations (e.g., resolution) on cameras and allocate network resources to reduce delay \cite{DBLP:journals/ton/ZhangWJWQXL22,DBLP:conf/nsdi/ZhangABPBF17,DBLP:journals/tmc/LuCPP19} and improve recognition accuracy \cite{DBLP:conf/infocom/LiuHOH18, DBLP:conf/sigcomm/JiangABSS18, DBLP:conf/infocom/WangWLJJ021}. While these works have made significant contributions, this paper seeks to address two open issues.

\begin{figure}[tbp] 
\centering 
\includegraphics[width=\linewidth]{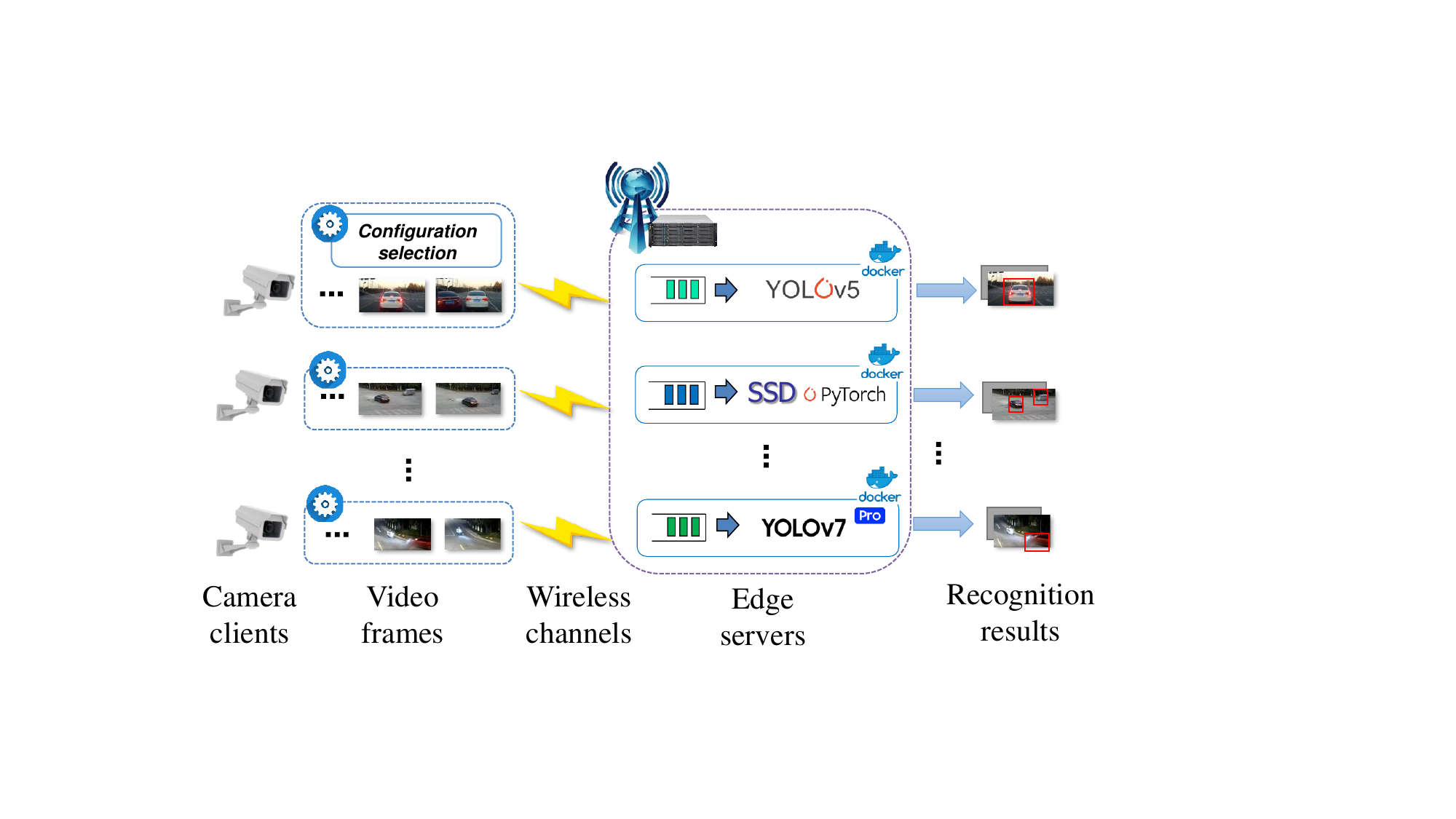} 
\caption{An illustration of the edge-assisted video analytics system.}
\label{Fig.system}
\end{figure}

Firstly, video analytics services require video frames to be transmitted in time and accurately recognized. Existing works usually focus on optimizing the Pareto-optimal or weighted sum of delay and accuracy \cite{DBLP:conf/infocom/LiuHOH18,DBLP:conf/nsdi/ZhangABPBF17}. Nevertheless, these heterogeneous metrics can be conflicting and incomparable, making it hard to determine the optimal trade-off, especially under time-varying network conditions. In this regard, a fundamental challenge is: \textit{How to model the overall performance of video analytics services in a unified manner, considering the integrated impact of transmission and computation phases?}

Secondly, the dynamic nature of video content induces the uncertainty of recognition accuracy, even if the video configuration remains unchanged. For example, uploading low-resolution frames achieves high accuracy if the detected objects move slowly, but the accuracy may significantly degrade in other cases. Additionally, the bandwidth and computation resources in edge networks are time-varying and usually unpredictable, affecting the transmission and computation efficiency. Given that video analytics services usually operate for a long period (e.g., illegal driving detection) \cite{DBLP:journals/pieee/ZhangSWZ19}, this raises another challenge: \textit{How to maintain the long-term performance of video analytics services without knowing the future variations of network conditions and video content?}

To handle the first challenge, we introduce the age of processed information (AoPI) concept in this work, which is the time elapsed since the generation of the latest accurately recognized frame received by the user. In the edge computing scenario, the widely-used delay metric only accounts for the time each frame takes to produce the recognition results. However, the delay metric ignores the generation interval between effective results of different frames, making the user's knowledge about the monitored environment outdated. On the other hand, the accuracy metric only accounts for the average recognition accuracy of all video frames, while the accumulative impact of recognition failures is ignored. As a variant of the widely-used age-of-information (AoI) concept \cite{DBLP:journals/tit/Yates20}, AoPI jointly reflects the transmission/computation efficiency and recognition accuracy. Lower AoPI indicates that users can be updated with the latest accurate results in time, which can be achieved with sufficient transmission/computation resources. Specifically, we derive the closed-form expressions of AoPI w.r.t the transmission, computation rate, and recognition accuracy of video frames under preemptive and non-preemptive computation policies at the edge server side, respectively. The derived results provide valuable insights for resource allocation and computation policy selection. The analytical results are further validated through testbed experiments.

To handle the second challenge, we formulate the joint problem of edge server selection, video configuration adaptation (i.e., video resolution, neural network model, and computation policy), and bandwidth/computation resource allocation among cameras. The formulated problem intends to minimize the long-term average AoPI over all cameras in an online manner, while meeting the predefined average recognition accuracy requirement. We propose an efficient Lyapunov-based block coordinate descent (LBCD) method to handle the formulated problem. LBCD first transforms the original long-term problem into a series of one-slot problems that require no future knowledge via Lyapunov optimization. Then, each one-slot problem is further decoupled into two subproblems. The first subproblem adjusts video configurations and allocates transmission/computation resources given the edge server selection strategy of cameras, which is solved using a block coordinate descent method. The second subproblem assigns cameras to different edge servers, which is solved using a first-fit method to match the cameras' resource requirements and the edge servers' capacity. We prove that the LBCD method can achieve asymptotically optimal performance. Testbed experiments and simulation results show that the LBCD is superior in resource-limited and large-scale scenarios. Specifically, LBCD can reduce the average AoPI by up to 10.94X compared with the state-of-the-art baselines.

Our main contributions in this paper are as follows.
\begin{itemize}
    \item The AoPI concept is adopted to model the overall performance of video analytics services, which depicts the integrated impact of transmission, computation efficiency, and accuracy. The closed-form expressions of AoPI are derived under preemptive and non-preemptive computation  policies, which are validated on the testbed.
    \item Aiming to minimize the long-term average AoPI of all cameras, the joint edge server selection, video configuration adaptation and bandwidth/computation resource allocation problem is formulated. The insight behind the problem is to design flexible strategies to accommodate timely video analytic services under the intrinsically dynamic video content and network variations.
    \item An online method, LBCD, is proposed to solve the long-term optimization problem, which can achieve provable asymptotically optimal performance without future information. Testbed experiments and simulation results show that LBCD can reduce the average AoPI by up to 10.94X, compared with the state-of-the-art baselines.
\end{itemize}

The rest of this paper is organized as follows. Section \ref{rel_wor} first reviews the related studies. Next, Section \ref{sys_mod} describes the system model and formulates the optimization problem. Then, the closed-form expressions of AoPI are derived in Section \ref{ave_aoi}. In Section \ref{oad}, an online method, LBCD, is proposed to handle the formulated problem. Section \ref{sim_res} presents the experimental results. Finally, in Section \ref{con_fut}, we conclude this paper and discuss future work.

\section{Related Work}
\label{rel_wor}
Video analytics services require delivering accurate results to users timely \cite{DBLP:journals/corr/abs-2211-15751}. A key issue in existing studies is how to adjust video configurations according to the user preferences and network conditions to improve the overall transmission and computation performance \cite{DBLP:journals/ton/ZhangWJWQXL22,DBLP:journals/tpds/KhochareKS21}. The existing works can be classified based on their different design objectives.

\textbf{Timeliness guarantee}: Recent studies pointed out that users care more about the timeliness than the average delay of the video analytics results, especially under dynamic network conditions \cite{DBLP:conf/infocom/zhangcasva22}. For example, in \cite{DBLP:conf/nsdi/ZhangABPBF17}, the \textit{lag} constraints of recognition results are considered to tune the configuration of different video tasks, i.e., the time elapsed since the arrival of the last processed frame should be lower than a certain threshold. In \cite{DBLP:conf/infocom/zhangcasva22}, the \textit{upload lag} is introduced to adjust the video configurations via reinforcement learning, i.e., the difference between the expected and actual uploading time of frames should be minimized. Besides these lag metrics, the AoI concept is more commonly used to depict the timeliness of real-time systems \cite{DBLP:journals/tit/Yates20}. Early studies on AoI focus on the sampling and transmission phase of data. Until recently, some researchers have studied the impact of the computation phase on AoI \cite{DBLP:journals/jsac/YatesSBKMU21a}. For example, the average AoI for processing computation-intensive messages with MEC under binary and partial offloading schemes is derived in closed-form in \cite{DBLP:journals/tvt/KuangGCM20}. In \cite{DBLP:journals/access/LiuQZZ20}, an algorithm with close-to-optimal AoI is designed to jointly optimize the offloading and computation resource allocation decisions.

\textbf{Accuracy enhancement}: Many existing studies focus on enhancing the inference accuracy of video analytics systems \cite{DBLP:conf/infocom/LiuHOH18}. For example, in \cite{DBLP:conf/sigcomm/JiangABSS18}, the temporal and spatial correlations in video contents are carefully utilized to help choose video configurations to achieve a higher accuracy. Given that the video contents are complex and time-varying, some researchers value the worst cases of frequent wrong recognition more than only the average accuracy performance. In \cite{DBLP:conf/infocom/WangWLJJ021}, an adaptive data rate controller is designed to improve the \textit{tail accuracy} of video analytics, i.e., the extremely low accuracy of some complex frames and object classes.

\textbf{Multi-objective optimization}: There are also many studies which optimize multiple objectives simultaneously \cite{DBLP:journals/corr/abs-2303-14329}. A typical example is \cite{DBLP:conf/infocom/LiuHOH18}, where Liu \textit{et al.} adjusted the video resolution to maximize the weighted sum of delay and inference accuracy under the minimum frame rate constraint. Similarly, in \cite{DBLP:journals/tii/LinYZLCY23}, Lin \textit{et al.} proposed a matching theory-based method to adjust the video frame rate to minimize the weighted sum of delay and accuracy of video analytics for on-road vehicles. Besides, the Pareto-optimal and hard/soft constraint of delay/lag and accuracy are also commonly considered in video configuration adaptation problem \cite{DBLP:journals/ton/ZhangWJWQXL22}. For instance, in \cite{DBLP:journals/iotj/WangHZZCCWC23}, the authors proposed an offloading framework to minimize the bandwidth usage by removing the spatial and temporal redundancies of target objects under the the accuracy and latency requirements of video analytics.

The existing works usually stop at the trade-off relationship between incomparable metrics like delay and accuracy. A gap exists between the overall system performance and the multi-dimensional inconsistent metrics. Motivated by this, we introduce an AoPI metric to model the transmission/computation efficiency and recognition accuracy jointly. Based on the derived AoPI, we study the joint edge server selection, video configuration adaptation and bandwidth/computation resource allocation problem.

\begin{table*}[t]
    \centering
    \renewcommand{\arraystretch}{1.25}
    \caption{Key notations}
    \label{tabel1}
    \begin{tabular}{|c|m{5.5cm}|c|m{5.5cm}|}
\hline
\textbf{Notation} & \textbf{Description} & \textbf{Notation} & \textbf{Description} \\ \hline
$y_{n,t}^s$ & whether camera \(n\) uploads frames to edge server $s$ in slot \(t\) & \(p_{n,t}\) & recognition accuracy of camera \(n\)'s video frames in slot \(t\) \\ \hline
\(b_{n,t}\) & bandwidth allocated to camera \(n\) by its edge server in slot \(t\) & \(A_{n,t}\) & average AoPI of camera \(n\) in slot \(t\) \\ \hline
\(r_{n,t}\) & video resolution selected by camera \(n\) in slot $t$ & $B_t^s$ & available bandwidth of edge server $s$ in slot $t$ \\ \hline
\(\lambda_{n,t}\) & average transmission rate of camera \(n\)'s frames in slot $t$ & $C_t^s$ & available computation capacity on edge server $s$ in slot $t$ \\ \hline
\(x_{n,t}\) & computation policy, i.e., FCFS or LCFSP, adopted by camera \(n\) in slot \(t\) & $P_{\mathrm{min} }$ & long-term average accuracy threshold \\ \hline
\(m_{n,t}\) & neural network model selected by camera \(n\) in slot \(t\) & $\bar{P} _t$ & average recognition accuracy of all cameras in slot $t$ \\ \hline
\(c_{n,t}\) & computation resource allocated to camera \(n\) by its edge server in slot \(t\) & $\bar{A} _t$ & average AoPI of all cameras in slot $t$ \\ \hline
\(\mu_{n,t}\) & average computation rate of camera \(n\)'s frames in slot $t$ & $V$ & adjustable weight parameter between average AoPI and recognition accuracy \\ \hline
\end{tabular}
\end{table*}

\section{System Model}
\label{sys_mod}
\textcolor{black}{
As depicted in Fig. 1, we consider a system that consists of multiple edge servers \(\mathcal{S} =\left \{ 1,\cdots ,S \right \} \) and a set of cameras \(\mathcal{N} =\left \{ 1,\cdots ,N \right \} \). These cameras are deployed in a certain area, e.g., an intersection or factory, and consistently upload captured videos to edge servers for analysis via multiple access points (such as LTE/5G/WiFi). The discrete time model is adopted and the system timeline consists of \(\mathcal{T} =\left \{ 1,\cdots,T \right \} \) time slots. Within each time slot, the available wireless bandwidth and computation capacity of each edge server remain stable. The time length of each slot is 5 minutes to adapt to the dynamic arrival of applications (e.g., traffic video analytics) and change of video content properties (e.g., weather and lighting). The main notations used in following sections are listed in Table \ref{tabel1}.}

\subsection{Frame Uploading Model}
\textcolor{black}{
Given that the available bandwidth and computation capacity of edge servers change with time due to workload fluctuations in edge networks, the cameras can dynamically select proper edge servers to process uploaded video frames. Denote by $y_{n,t}^s=1 $ if edge server $s$ is selected by camera \(n\) in slot \(t\), and $y_{n,t}^s=0 $ otherwise.}

\textcolor{black}{
During each time slot, every camera uploads a new frame to its edge server when the previous frame’s transmission finishes. We adopt this mechanism to avoid the unnecessary idle time on the cameras during the frame transmission process. Besides, the camera may generate frames at a high rate, e.g., 30 FPS. In such case, only the latest frame is uploaded because it contains the latest information and exhibits the most significant divergence from the previously uploaded frames \footnote{The waiting delay of uploaded frames is negligible (several milliseconds) and thus omitted in following analysis.}.} Denote by \(b_{n,t}\) (in Hz) the bandwidth allocated to camera \(n\) by its edge server in slot \(t\). Then the average transmission data rate of camera \(n\) in slot \(t\) is given by:
\begin{equation}
w_{n,t}=b_{n,t}\mathrm{log} _2\left ( 1+\frac{\tilde{E} _n \tilde{G} _{n}}{\sigma }  \right ) ,
\end{equation}
where \(\tilde{E} _n\), \(\tilde{G} _{n}\) and \(\sigma\) represent the transmission power, channel gain of camera \(n\) and the noise power, respectively.

\textcolor{black}{
Denote by \(\mathcal{R} \) the set of all alternative video resolutions. For camera \(n\) in slot $t$, the data size of a frame with resolution \(r_{n,t} \in \mathcal{R}\) is \(\alpha r_{n,t}^2\) (in bits) on average, where \(\alpha\) is a constant parameter depending on the encoding techniques \cite{H.264}. In the wireless environment, the data packets usually have a certain probability of encountering transmission failure due to channel fading and collision in the physical layer\cite{DBLP:journals/icl/ZhangWKL16}. Then, the delay of successfully transmitting camera \(n\)'s frames in slot $t$ can be modeled as an exponentially distributed variable with a mean value \(1/\lambda_{n,t}\)\cite{DBLP:journals/jsac/SriramW86}.} Specifically, \(\lambda _{n,t}\) is given by:
\begin{equation}
\lambda _{n,t}=\frac{w_{n,t}}{\alpha r_{n,t}^2}.
\end{equation}

\subsection{Frame Processing Model}
The edge servers assign different containers to cameras for  recognition. In each container, either first-come-first-serve (FCFS) or last-come-first-serve-with-preemption (LCFSP) computation policy is adopted. Newly-arrived frames have to wait in a queue before recognition under FCFS policy. With LCFSP policy, new frames can preempt under-processed frames on the edge server. Denote by \(x_{n,t}\in \left \{ 0,1 \right \} \) the computation policy adopted by camera \(n\) in slot \(t\), where \(x_{n,t}=0\) represents the FCFS policy and \(x_{n,t}=1\) represents the LCFSP policy.

\textcolor{black}{
The edge servers mainly handle object detection and segmentation tasks, which focus on processing independent frames. These tasks are widely used in practical systems, e.g., alerting traffic accidents or monitoring faults on a production line\cite{DBLP:journals/comsur/JedariPIFMY21}. There are also many complex tasks in video analytics applications, such as behavior analysis and visual question answering, which usually have more lenient timing requirements and rely on the pre-processing results of object detection and segmentation tasks\cite{DBLP:journals/corr/abs-2303-14329}. Hence, in this work, we focus on the fundamental detection/segmentation tasks and take the AoPI optimization for these complex tasks as the future work.}

The cameras can select different neural network models to process video frames, e.g., YOLOv5\cite{yolov5m}, SSD\cite{DBLP:conf/eccv/LiuAESRFB16}, etc. The computation requirement and recognition accuracy of these models usually vary significantly in practical systems (e.g., up to 50X difference) \cite{DBLP:conf/usenix/Romero0YK21}. Denote by \(\mathcal{M}\) the set of all alternative neural network models. For camera \(n\), the computational complexity for processing a frame with resolution \(r_{n,t}\) using model \(m_{n,t} \in \mathcal{M}\) is \(\xi \left ( r_{n,t}, m_{n,t} \right ) \) (in FLOPs) on average, where \(\xi \left ( \cdot, \cdot \right )  \) is a convex function w.r.t frame resolution \cite{DBLP:conf/infocom/LiuHOH18} and proportional to the model parameters number \cite{DBLP:conf/iclr/MolchanovTKAK17}.

\textcolor{black}{
Denote by \(c_{n,t}\) (in FLOPS) the computation resource allocated to camera \(n\)'s container by its edge server in slot \(t\). In practical systems, the processing latency for frames varies depending on the video content\cite{padilla2021comparative}. For instance, frames containing many small objects require more time for the neural network model to process. However, such frames have a relatively low probability of occurrence. Previous computer vision studies found that the exponential distribution can  properly depict such characteristics\cite{DBLP:conf/iccv/MaoYD19}. Hence, the computation delay of camera \(n\)'s frames in slot $t$ is modeled as an exponentially distributed variable with a mean value \(1/\mu_{n,t}\).} Specifically, \(\mu_{n,t}\) is given by:
\begin{equation}
\mu_{n,t}=\frac{c_{n,t}}{\xi \left ( r_{n,t}, m_{n,t} \right )}.
\end{equation}

\textcolor{black}{
The recognition of frames may be either accurate or inaccurate  depending on the video resolution, contents, neural network model, and the specific computer vision tasks. For instance, in object detection tasks, a video frame is considered accurately recognized if the bounding boxes in the results possess the correct label and exhibit sufficient spatial overlap with the ground truth\cite{DBLP:journals/pieee/ZouCSGY23}. In image semantic segmentation tasks, a video frame is deemed accurately recognized if an adequate number of pixels in the results share the same label as the ground truth\cite{DBLP:journals/ijon/HaoZG20}.} 

Denote by \(p_{n,t}\) the recognition accuracy of camera \(n\)'s video frames in slot \(t\). We have:
\begin{equation}
    p_{n,t}=\zeta_n^t \left ( r_{n,t }, m_{n,t}\right ) ,
\end{equation}
where $\zeta_n^t \left ( \cdot, \cdot \right )$ is a concave and monotone increasing function w.r.t frame resolution and model parameters number\cite{DBLP:journals/ton/ZhangWJWQXL22}. Note that $\zeta_n^t \left ( \cdot, \cdot \right )$ is impacted by video content and, thus, changes with time. \textcolor{black}{ In practical systems, $\zeta_n^t \left ( \cdot, \cdot \right )$ can be profiled at the beginning of slot $t$ with low overhead \cite{DBLP:conf/cloud/RomeroZYK21}. Specifically, the latest profiling techniques can finish the profiling process within 10 seconds\cite{DBLP:conf/iwqos/Wu0MZQWZC21}. Besides, the profiling results of previous time slots and nearby cameras can be reused based on comparing the data distributions of different videos \cite{DBLP:journals/access/KimY20b}.}

After recognition, the results are transmitted to the users or the cloud for future usage via the backhaul network, e.g., person re-identification task which fuses the recognition results from different cameras \cite{DBLP:journals/pami/LuoSZ23}. Nevertheless, the unstable frame uploading and recognition phases are the main bottleneck in video analytics services, which are the main focus of this work.

\textcolor{black}{
Denote by \(A_{n,t}\) the average AoPI of camera \(n\) in slot \(t\). The analysis of AoPI relies on the exponentially distributed frame transmission and computation delay. Although the real-world delays may not adhere to exponential distributions, this does not invalidate our proposed method. In our testbed experiments, we found that the transmission and computation delays are more evenly distributed than the exponential distributions. However, the proposed method still outperforms state-of-the-art baselines across various computer vision tasks, illustrating its great potential.}

\subsection{Problem Formulation}

Given the closed-form expression of \(A_{n,t}\) derived in Section \ref{ave_aoi}, we aim to minimize the long-term average AoPI of all cameras. Besides, practical video analytics services usually operate for a long period (e.g., illegal driving detection) and have minimum accuracy constraints on recognition results \cite{DBLP:journals/pieee/ZhangSWZ19}. Hence, we require the long-term average accuracy across all cameras to exceed a certain threshold $P_{\mathrm{min} }$. Then, the optimization problem is:

\begin{align}
\mathbf{\left ( P1 \right )} \ \min_{\left \{ \mathbf{y},\mathbf{r},\mathbf{x},\atop \mathbf{m}, \mathbf{b},\mathbf{c} \right \} } &  \lim_{T \to \infty} \frac{1}{T} \sum_{t = 1}^{T} \frac{1}{N} \sum_{n  = 1}^{N} A_{n,t} \\
s.t. & \ \sum_{n=1}^{N} y_{n,t}^s b_{n,t} \le B_t^s, \forall s \in \mathcal{S}, \forall  t \in \mathcal{T} \\
& \ \sum_{n=1}^{N} y_{n,t}^s c_{n,t} \le C_t^s, \forall s \in \mathcal{S}, \forall  t \in \mathcal{T} \\
& \ \sum_{n=1}^{N} y_{n,t}^s = 1, \forall s \in \mathcal{S}, \forall  t \in \mathcal{T} \\
& \ \lim_{T \to \infty} \frac{1}{T} \sum_{t=1}^{T} \frac{1}{N} \sum_{n=1}^{N} p_{n,t} \ge P_{\mathrm{min} } \\
& \ \lambda_{n,t} (1-x_{n,t})< \mu_{n,t}, \forall  n \in \mathcal{N}, \forall  t \in \mathcal{T}
\end{align}

The decision variables in problem (P1) include the edge server selection strategy $y_{n,t}^s$, the adopted video configuration, the bandwidth allocation strategy $b_{n,t}$ and the computation resource allocation strategy $c_{n,t}$. Specifically, the video configuration consists of the video resolution $r_{n,t}$, the computation policy $x_{n,t}$ and the neural network model $m_{n,t}$.

In problem (P1), constraint (6) represents that the bandwidth allocated to cameras cannot exceed the available bandwidth $B_t^s$ of edge server $s$ in slot $t$. Constraint (7) represents the limitation of computation capacity $C_t^s$ on the edge server $s$ in slot $t$. Constraint (8) means that every camera should select one edge server to process video frames. Constraint (9) means that the average accuracy of all cameras should be higher than $P_{\mathrm{min} }$ in the long term. Constraint (10) means that the transmission rate of frames should be lower than the computation rate under the FCFS policy; otherwise, a long queuing delay occurs.

\textcolor{black}{
According to Section \ref{ave_aoi}, (P1) is an NP-hard mixed integer nonlinear programming problem, which is notoriously hard to solve. Consider a small system that consists of 10 cameras, 2 edge servers, 5 resolution candidate and 3 model candidate, the solution space size in this case is still up $O\left ( 30^{10T} \right ) $ even ignoring the resource allocation decisions. The commonly-used offline methods are not applicable to handle such a large solution space. Moreover, the lack of future information about the network condition and video content further complicates the problem. Therefore, an efficient online method is needed.}

In Section \ref{oad}, we present the LBCD method to handle the problem with asymptotically optimal performance. LBCD first transforms (P1) into many one-slot problems via Lyapunov framework, which enables real-time decision making with no future knowledge. LBCD then decouples each one-slot problem into configuration adaptation/resource allocation subproblem and edge server selection subproblem, which are efficiently solved via block coordinate descent and first-fit method, respectively.

\section{AoPI Analysis}
\label{ave_aoi}
In this section, we derive the average AoPI of camera \(n\) in slot \(t\) under FCFS and LCFSP policy. The subscript in \(\lambda_{n,t}\), \(\mu_{n,t}\) and \(p_{n,t}\) is removed for notation simplicity.

\begin{figure}[t]
\centering 
\includegraphics[width=0.8\linewidth]{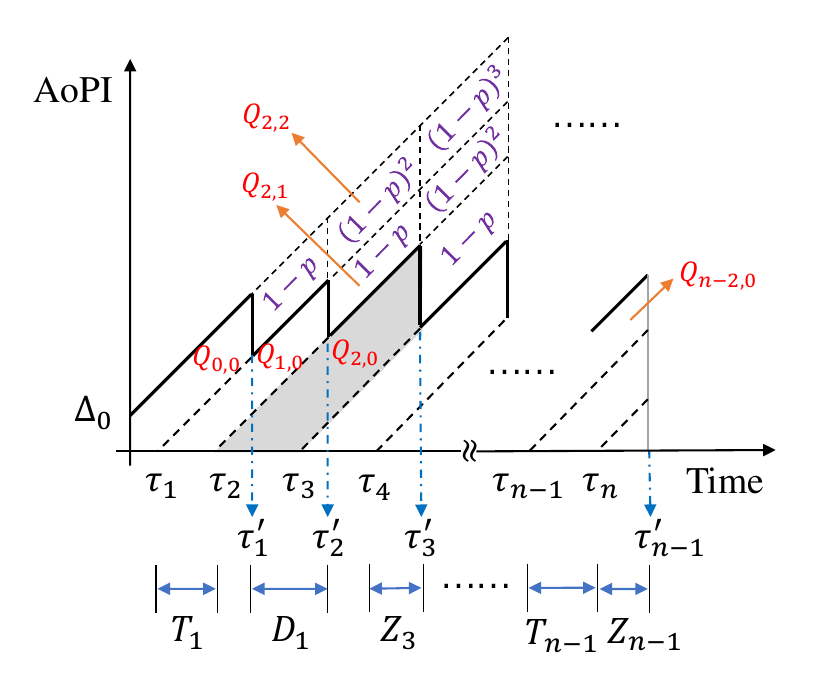} 
\caption{AoPI evolution under FCFS computation policy.}
\label{Fig.fcfs}
\end{figure}

\subsection{Average AoPI under FCFS Policy}
\begin{theorem}
If the transmission and computation delay of video frames follow exponential distributions with mean value \(1/\lambda\) and \(1/\mu\) respectively, the average AoPI of camera \(n\) in slot \(t\) under FCFS policy is:
\begin{equation}
A_\mathrm{F}=\left ( 1+\frac{1}{p}  \right ) \frac{1}{\lambda } +\frac{1}{\mu} +\frac{2\lambda ^3+\lambda \mu^2-\mu\lambda ^2}{\mu^4-\mu^2\lambda ^2} ,
\end{equation}
where \(p\) denotes the recognition accuracy.
\end{theorem}
\begin{proof}
The AoPI evolution process under the FCFS policy is presented in Fig. \ref{Fig.fcfs}. Denote by \(\tau _i\) and \(\tau '_i\) the time instant of frame \(i\)'s generation at the camera and computation completion at the edge server, respectively. Denote by \(T_i\) and \(O_i\) the transmission and computation time of frame \(i\), respectively. Denote by \(Z_i\) the time that frame \(i\) stays in the edge server, \(Z_i=\tau '_{i}-\tau _{i+1}\). Denote by \(D_i=\tau '_{i+1}-\tau '_i\) the inter-departure time between frame \(i\) and \(i+1\) leaves the edge server.

There are multiple age evolution curves under the FCFS policy, as illustrated in Fig. \ref{Fig.fcfs}, which depends on whether the previous frames are accurately recognized or not. For convenience, we divide the whole area under all age evolution curves into different segments. Denote by \(Q_{i,0} \) the AoPI evolution area brought by accurately recognizing frame \(i\). Denote by \(Q_{i,j}\) the age evolution areas brought by inaccurately recognizing frames \(i-j+1,\cdots,i\). Hence:
\begin{equation}
\label{age_q_f}
Q_{i,0}=\frac{1}{2}T_i^2+T_i T_{i+1}+T_i Z_{i+1},
\end{equation}

\begin{equation}
\label{age_qq_f}
Q_{i,j}=D_i T_{i-j}.
\end{equation}

Without loss of generality, we assume the recognition results of different frames are independent. Then, under the FCFS policy, the segment \(Q_{i,j}\) exists with probability \((1-p)^{j}\). Thus, the accumulative AoPI from time 0 to \( \tau'_{n-1}\) is:
\begin{equation}
Q_\mathrm{F} =Q_{0,0}+\sum_{i=1}^{n-2} \sum_{j=0}^{i} (1-p)^{j}Q_{i,j}  +\frac{1}{2} \left ( T_{n-1} +Z_{n-1}\right )^2.
\end{equation}

Then, the average AoPI from time 0 to \( \tau'_{n-1}\) is:
\begin{equation}
\label{keyeqf}
\begin{split}
A_\mathrm{F}  =& \lim_{n \to \infty} \frac{Q_\mathrm{F} }{\tau'_{n-1}}
\\=&\lambda
 \lim_{n \to \infty}\frac{1}{n-2}\sum_{i=1}^{n-2}\sum_{j=1}^{i}(1-p)^j\mathbb{E}\left [ D_i T_{i-j} \right ] \\&+\lambda \left ( \frac{1}{2} \mathbb{E} \left [ T_i^2 \right ]+ \mathbb{E} \left [ T_i T_{i+1} \right ]+\mathbb{E} \left [ T_i Z_{i+1} \right ] \right ).
\end{split}
\end{equation}

The transmission time of frames follows exponential distributions. Hence, we have:
\begin{equation}
\label{trans_sq}
\mathbb{E} \left [ T_i^2 \right ] =\frac{2}{\lambda ^2} , \mathbb{E} \left [ T_i T_{i+1} \right ] =\mathbb{E} \left [ T_i \right ] \mathbb{E} \left [ T_{i+1} \right ]=\frac{1}{\lambda ^2} .
\end{equation}

Note that \(Z_{i+1}\) consists of the waiting time \(W_{i+1}\) and the computation time \(O_{i+1}\) of frame \(i+1\) on the edge server. Thus, \(\mathbb{E} \left [ T_i Z_{i+1} \right ]=\mathbb{E} \left [ T_i W_{i+1} \right ]+\mathbb{E} \left [ T_i O_{i+1} \right ]\). Recall that the computation time of frames follow exponential distributions. Hence, we have:
\begin{equation}
\label{transpro}
\mathbb{E} \left [ T_i O_{i+1} \right ]=\mathbb{E} \left [ T_i\right ]\mathbb{E} \left [ O_{i+1} \right ]=\frac{1}{\lambda\mu}.
\end{equation}

If frame \(i+1\) arrives at the edge server after \(i\)'s computation finishes, then frame \(i+1\) can be processed immediately and the waiting time \(W_{i+1}=0\). Otherwise, the waiting time of frame \(i+1\) is \(W_{i+1}=Z_{i}-T_{i+1}\). Thus:
\begin{equation}
\begin{split}
W_{i+1}=&\left ( Z_i-T_{i+1} \right ) ^+
\\=&\left ( W_i+O_i-T_{i+1} \right ) ^+
\\=&\left ( \left (Z_{i-1}-T_i \right )^++ O_i-T_{i+1}  \right ) ^+.
\end{split}
\end{equation}
Note that \(Z_{i-1}\) is independent of \(T_i\), \(T_{i+1}\) and \(O_i\). Therefore, the conditional expectation of \(W_{i+1}\) given \(T_i=\tau'\) is:
\begin{equation}
\begin{split}
\mathbb{E} &\left [ W_{i+1}|T_i=\tau' \right ] 
\\=&\iiint_{0}^{\infty }  f_Z(z) f_O(o) f_T(\tau)\left ( \left ( z-\tau' \right )^++o-\tau  \right ) ^+dz do d\tau
\\=&\frac{2\lambda }{(\mu+\lambda)(\mu-\lambda )} e^{(\lambda -\mu )\tau'}+\frac{\lambda }{\mu (\lambda +\mu )} ,
\end{split}
\end{equation}
where \( f_Z(z), f_O(o), f_T(\tau)\) denote the probability density function (PDF) of \(Z_i\), \(O_i\), and \(T_i\), respectively.

Next, we have:
\begin{equation}
\label{transwait}
\begin{split}
\mathbb{E}[T_i W_{i+1}]=&\int_{0}^{\infty } \tau f_T(\tau) \mathbb{E}[W_{i+1}|T_i=\tau]d\tau
\\=&\frac{2\lambda ^2+\mu ^2-\lambda \mu}{\mu^4-\mu^2\lambda ^2}.
\end{split}
\end{equation}

When the M/M/1 queue on the edge server reaches steady-state, the inter-arrival time \(T_{i-j}, 1\le j\le i\) is independent of the inter-departure time \(D_i\) \cite{DBLP:journals/corr/Zukerman13}. Therefore:
\begin{equation}
\label{trans_dep_mul}
\mathbb{E}\left [ D_i T_{i-j} \right ]=\mathbb{E}\left [ D_i\right ]\mathbb{E}\left [ T_{i-j} \right ]=\frac{1}{\lambda ^2} .
\end{equation}

Combine Eqs. (\ref{trans_sq}), (\ref{transpro}), (\ref{transwait}) and (\ref{trans_dep_mul}) with (\ref{keyeqf}), the average AoPI under the FCFS policy is:
\begin{equation}
\begin{split}
A _\mathrm{F} =\left (1+ \frac{1}{p} \right )\frac{1}{\lambda } +\frac{1}{\mu } + \frac{2\lambda ^3+\lambda \mu ^2-\mu\lambda^2 }{\mu^4-\mu^2\lambda ^2}.
\end{split}
\end{equation}
Theorem 1 is thus proved.
\end{proof}

According to Theorem 1, the AoPI under the FCFS policy is inversely proportional to the recognition accuracy. The influence of the transmission and computation rate are discussed as follows.

\begin{itemize}
    \item \textbf{Impact of transmission rate}
\end{itemize}

\begin{corollary} 
\(A _\mathrm{F} \) is a convex function with respect to \(\lambda\), which first decreases and then increases with \(\lambda\).
\label{cor-1}
\end{corollary}

Corollary \ref{cor-1} can be proved by analyzing the derivatives of \(A _\mathrm{F} \). From corollary \ref{cor-1}, the transmission rate should be neither too low nor too high. Otherwise, the AoPI performance will degrade due to the long time interval between received recognition results or the high queuing delay.

Although the closed-form expression of optimal transmission rate \(\lambda^*\) is hard to obtain, we can prove that \(\lambda^*\) always decreases with the recognition accuracy \(p\). The insight is that more frames should be uploaded for compensation in the cases of low recognition accuracy. However, such compensation cannot maintain the AoPI level with the decline of \(p\).

\begin{itemize}
    \item \textbf{Impact of computation rate}
\end{itemize}

\begin{corollary} 
\(A _\mathrm{F} \) is a convex and monotone decreasing 
function with respect to \(\mu\).
\label{cor-3}
\end{corollary}

Corollary \ref{cor-3} can be proved by analyzing the the derivatives of \(A _\mathrm{F} \). The insight is that reducing the computation delay of frames always leads to a lower AoPI.

\begin{itemize}
    \item \textbf{Transmission-computation resources trading}
\end{itemize}

\begin{figure}[tbp]
	\centering
	\begin{subfigure}{0.48\linewidth}
		\centering
		\includegraphics[width=0.95\linewidth]{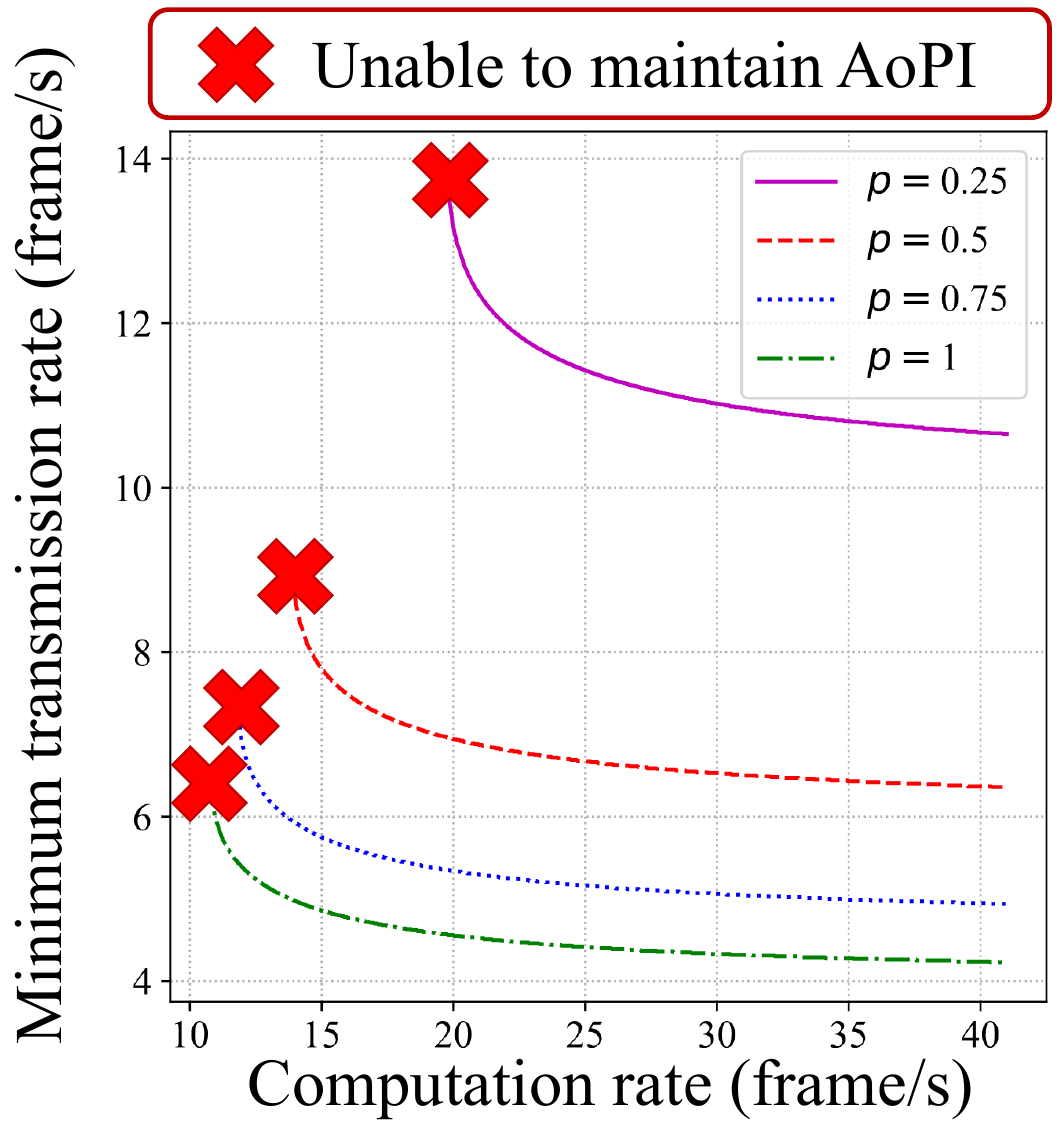}
        \caption{}
		\label{fcfs_trade_t}
	\end{subfigure}
	\centering
	\begin{subfigure}{0.48\linewidth}
		\centering
		\includegraphics[width=0.95\linewidth]{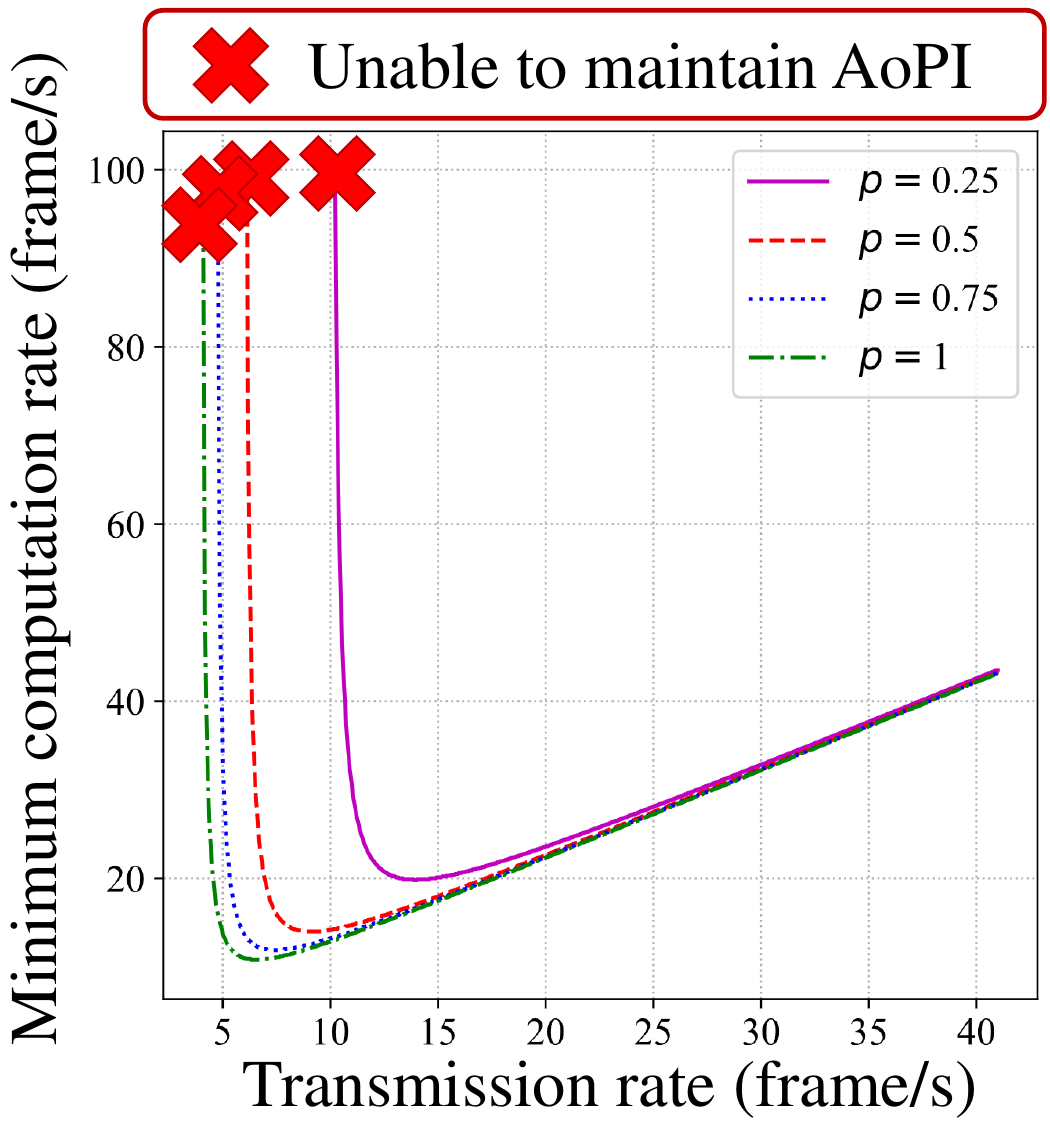}
        \caption{}
		\label{fcfs_trade_c}
	\end{subfigure}
	\caption{Minimum transmission rate (a) and computation rate (b) required to keep the average AoPI lower than 0.5s under the FCFS policy.}
	\label{fcfs_trade}
\end{figure}

Another interesting question is the minimum transmission and computation rate to satisfy a given average AoPI constraint. As illustrated in Fig. \ref{fcfs_trade}, the minimum transmission rate continuously decreases with the preserved computation rate. However, the minimum computation rate first decreases and then increases with the preserved transmission rate. The main reason is that the computation rate should increase in case of a high transmission rate to avoid an infinite queuing delay on the edge server.

\subsection{Average AoPI under LCFSP Policy}
\begin{theorem}
If the transmission and computation time of video frames follow exponential distributions with mean value \(1/\lambda\) and \(1/\mu\) respectively, the average AoPI of camera \(n\) in slot \(t\) under LCFSP policy is:
\begin{equation}
A_\mathrm{L}=\left ( 1+\frac{1}{p} \right ) \frac{1}{\lambda } +\frac{1}{p\mu} ,
\end{equation}
where \(p\) represents the recognition accuracy.
\end{theorem}
\begin{proof}
Compared with the FCFS strategy, the LCFSP strategy allows preemption in service, i.e., only some frames can finish computation on the edge server. The AoPI evolution process under the LCFSP policy is presented in Fig. \ref{Fig.lcfsp}, where each symbol {\color{red} x } represents a preemption happens. 

We denote by \(\tau_i\) the generation time of the \(i\)-th frame that finishes the computation. The inter-generation time, \(G_i=\tau_{i+1}-\tau_{i}\), may contain the transmission time of multiple preempted frames, which is difficult to analyze. Hence, we focus on the inter-departure time \(Y_i=\tau'_{i+1}-\tau'_i\), where \(\tau'_i\) denotes the computation finishing time instant of the \(i\)-th frame that is not preempted. Denote by \(T_i\) and \(O_i\) the transmission and computation time of the \(i\)-th frame, respectively. Besides, we denote by \(L_i=\tau'_i-\tau_i\) the time between the \(i\)-th frame's generation and computation finishing.

The whole area under the age evolution curves is divided into different segments \(Q_{i,j}\) as depicted in Fig. \ref{Fig.lcfsp}, which follows the same method as in the proof of Theorem 1. We have:
\begin{equation}
\label{age_q_l}
Q_{i,0}=\frac{1}{2} \left ( L_i+Y_i \right ) ^2-\frac{1}{2} L_{i+1}^2,
\end{equation}
\begin{equation}
\label{age_qq_l}
\begin{split}
Q_{i,j}=\left ( L_{i-j}+Y_{i-j}-L_{i-j+1} \right ) Y_i.
\end{split}
\end{equation}

\begin{figure}[t] 
\centering 
\includegraphics[width=0.8\linewidth]{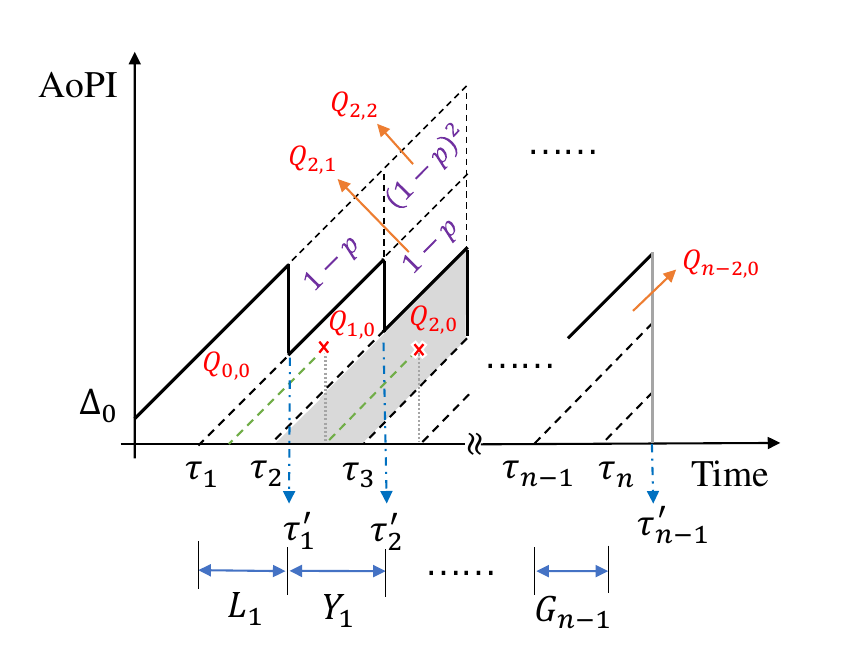} 
\caption{AoPI evolution under LCFSP computation policy.}
\label{Fig.lcfsp}
\end{figure}

In general, \(Q_{i,j}\) exists with probability \((1-p)^j\). Therefore, the accumulative AoPI from time 0 to \( \tau'_{n-1}\) is:
\begin{equation}
Q_\mathrm{L}=Q_{0,0}+\sum_{i=1}^{n-2}\sum_{j=0}^{i}(1-p)^j Q_{i,j} +\frac{1}{2} L_{n-1}^2.
\end{equation}
Then, the average AoPI from time 0 to \( \tau'_{n-1}\) is:
\begin{equation}
\label{ave_aoi_l}
\begin{split}
A_\mathrm{L}=&\lim_{n \to \infty} \frac{Q_\mathrm{L}}{\tau'_{n-1} } 
\\=&\lambda _e\left ( \mathbb{E} [L_i Y_i] +\frac{1}{2}\mathbb{E} [Y_i^2] \right ) \\&+\lambda _e\lim_{n \to \infty}\sum_{i = 1}^{n-2}\sum_{j = 1}^{i} (1-p)^{j} \mathbb{E} [G_{i-j}Y_i].
\end{split}
\end{equation}
where \(\lambda _e\) represents the effective arrival rate \cite{DBLP:journals/tit/CostaCE16}:
\begin{equation}
\label{eff_rat}
\begin{split}
\lambda _e=\lambda \cdot \mathbb{P} \left ( \left \{ \mathrm{not\ preempt\ frame}  \right \}  \right ) =\frac{\lambda\mu }{\lambda +\mu}.
\end{split}
\end{equation}

Next, we need to calculate \(\mathbb{E} [L_i Y_i]\), \(\mathbb{E} [Y_i^2]\) and \(\mathbb{E} [G_{i-j}Y_i]\). Firstly, we have:
\begin{equation}
\mathbb{E} [L_i Y_i]=\mathbb{E} [L_i]\mathbb{E} [Y_i],
\end{equation}
which is due to \(L_i\) and \(Y_i\) are independent. To prove this, we denote by \(H_i\) the time interval between frame \(i\)'s departure and the next generated frame's arrival at the edge server. Denote by \(R_i\) the remaining time for a new frame to complete computation. Note that the inter-departure time \(Y_i=H_i+R_i\). According to the definition, \(R_i\) is independent of \(T_i\) because there is no overlap between them. Moreover, the time \(H_i\) is also independent of \(L_i\) (please refer to \cite{DBLP:conf/isit/NajmN16}, Section 3-A).

Now we turn to the calculation of \(\mathbb{E} [L_i]\). We have:
\begin{equation}
\label{eci}
\mathbb{E} [L_i]=\mathbb{E} [T_i+O_i|O_i<T']=\mathbb{E} [T_i]+\mathbb{E} [O_i|O_i<T'],
\end{equation}
where the conditional expectation indicates the $i$-th frame's computation time should be smaller than the transmission time of the next frame to avoid preemption. We have:
\begin{equation}
\mathbb{E} [O_i|O_i<T']=\int_{0}^{\infty }o\cdot f_{O|O<T}(o) do ,
\end{equation}
and
\begin{equation}
\label{con_exp}
\begin{split}
&f_{O|O<T}(o)
\\&\overset{(a)}{=} \lim_{\eta  \to 0}\frac{\mathbb{P} (o\le O<o+\eta)\mathbb{P} (O<T|o\le O<o+\eta)}{\eta \mathbb{P} (O<T)}
\\&=(\lambda +\mu )e^{-(\lambda +\mu)o},
\end{split}
\end{equation}
where step (a) is derived according to Bayes rule. Therefore:
\begin{equation}
\label{sys_time}
\mathbb{E} [O_i|O_i<T']=\frac{1}{\lambda +\mu }.
\end{equation}

Then we need to calculate \(\mathbb{E} [Y_i]\) based on the moment generating function of \(Y_i\), which consists of the transmission time of multiple preempted frames and the computation time of the final frame. Similar to Eq. (\ref{con_exp}), the probability density function of the preempted frames' transmission time \(\hat{T}\) is:
\begin{equation}
\begin{split}
f_{\hat{T}}(\tau)=f_{T|T<O}(\tau)=(\lambda +\mu)e^{-(\lambda +\mu)\tau}
\end{split},
\end{equation}
and thus the moment generating function of \(\hat{T}\) is:
\begin{equation}
\label{mgft}
\phi _{\hat{T} }(s)=\int_{0}^{\infty }f_{\hat{T}}(\tau)e^{s\tau}d\tau=\frac{\lambda +\mu }{\lambda +\mu -s}  .
\end{equation}

Similarly, the computation time \(\hat{O}\) of the final frame in \(Y_i\) should be smaller than the transmission time of the next frame. Thus, the moment generating function of \(\hat{O}\) is:
\begin{equation}
\label{mgfp}
\phi _{\hat{O} }(s)=\int_{0}^{\infty }f_{\hat{O}}(o)e^{so}do=\frac{\lambda +\mu }{\lambda +\mu -s}  .
\end{equation}

Given Eqs. (\ref{mgft}) and (\ref{mgfp}), we can calculate \(Y_i\)'s moment generating function as follows. Note that the probability that a frame is not preempted is \(p_c=\frac{\mu}{\lambda +\mu} \) according to Eq. (\ref{eff_rat}). Denote by \(\tilde{T}\) the transmission time of the first frame generated after frame \(i\). Thus, we have:
\begin{equation}
Y_i=\tilde{T} +\sum_{i=0}^{K} \hat{T}_i +\hat{O} ,
\end{equation}
where \(K\) represents the number of preempted frames between departure of frame \(i\) and \(i+1\). Note that \(K\) is a random variable that follows geometric distribution and the moment generating function of \(\tilde{T} \) is \(\phi _{\tilde{T} }=\frac{\lambda  }{\lambda -s} \). Besides, \(K\), \(\tilde{T} \), \(\hat{T}\), and \(\hat{P}\) are all mutually independent. Therefore, the moment generating function of \(Y_i\) is:
\begin{equation}
\begin{split}
\phi _{Y_i}(s)&=\mathbb{E}[e^{s(\tilde{T} +\sum_{i=0}^{K} \hat{T}_i +\hat{O})}]
\\&=\mathbb{E}[e^{s\sum_{i=0}^{K} \hat{T}_i}]\phi _{\tilde{T} }(s)\phi _{\hat{O} }(s)
\\&=\sum_{i=0}^{\infty }p_c(1-p_c)^i \phi _{\hat{T} }(s)^i\phi _{\tilde{T} }(s)\phi _{\hat{O} }(s)
\\&=\frac{\lambda \mu}{(\lambda -s)(\mu -s)}.
\end{split}
\end{equation}

Given \(\phi _{Y_i}(s)\), we have:
\begin{equation}
\label{epit}
\mathbb{E} [Y_i]=\phi _{Y_i}'(s)|_{s=0}=\frac{\lambda+ \mu}{\lambda \mu} ,
\end{equation}

\begin{equation}
\label{dep_sq}
\mathbb{E} [Y_i^2]=\phi _{Y_i}''(s)|_{s=0}=\frac{2(\lambda+ \mu)}{\lambda^2 \mu^2} -\frac{2}{\lambda \mu} .
\end{equation}

Next, we need to consider \(\mathbb{E} [G_{i-j}Y_i]\). Because \(G_{i-j}=T_{i-j}+Y_{i-j}-T_{i-j+1}\), we have:
\begin{equation}
\label{ge_de_mul}
\mathbb{E} [G_{i-j}Y_i]=\mathbb{E} [Y_{i-j}Y_i]=\mathbb{E} [Y_i]^2,
\end{equation}
which is due to \(T_{i-j}\) and \(Y_{i-j}\) are all independent of \(Y_i\). 

Finally, combine Eqs. (\ref{eff_rat}), (\ref{eci}), (\ref{sys_time}), (\ref{epit}), (\ref{dep_sq}) and (\ref{ge_de_mul}) with (\ref{ave_aoi_l}), the average AoPI under the LCFSP policy is:
\begin{equation}
\begin{split}
A _\mathrm{L}=\left ( 1+\frac{1}{p}  \right ) \frac{1}{\lambda } +\frac{1}{p\mu }.
\end{split}
\end{equation}

Theorem 2 is thus proved.
\end{proof}

According to Theorem 2, the AoPI under the LCFSP policy is inversely proportional to the recognition accuracy \(p\), transmission rate \(\lambda\) and computation rate \(\mu\). Thus, increasing the transmission and computation rate always reduces the average AoPI. 

For the predefined average AoPI requirement, the required minimum transmission and computation rate is shown in Fig. \ref{lcfsp_trade}. As can be seen, the minimum transmission (resp. computation) rate always decreases with the preserved computation (resp. transmission) rate. The main reason is that both high transmission and computation rate help preempt processing obsolete data under the LCFSP policy.

\begin{figure}[tbp]
	\centering
	\begin{subfigure}{0.48\linewidth}
		\centering
		\includegraphics[width=0.95\linewidth]{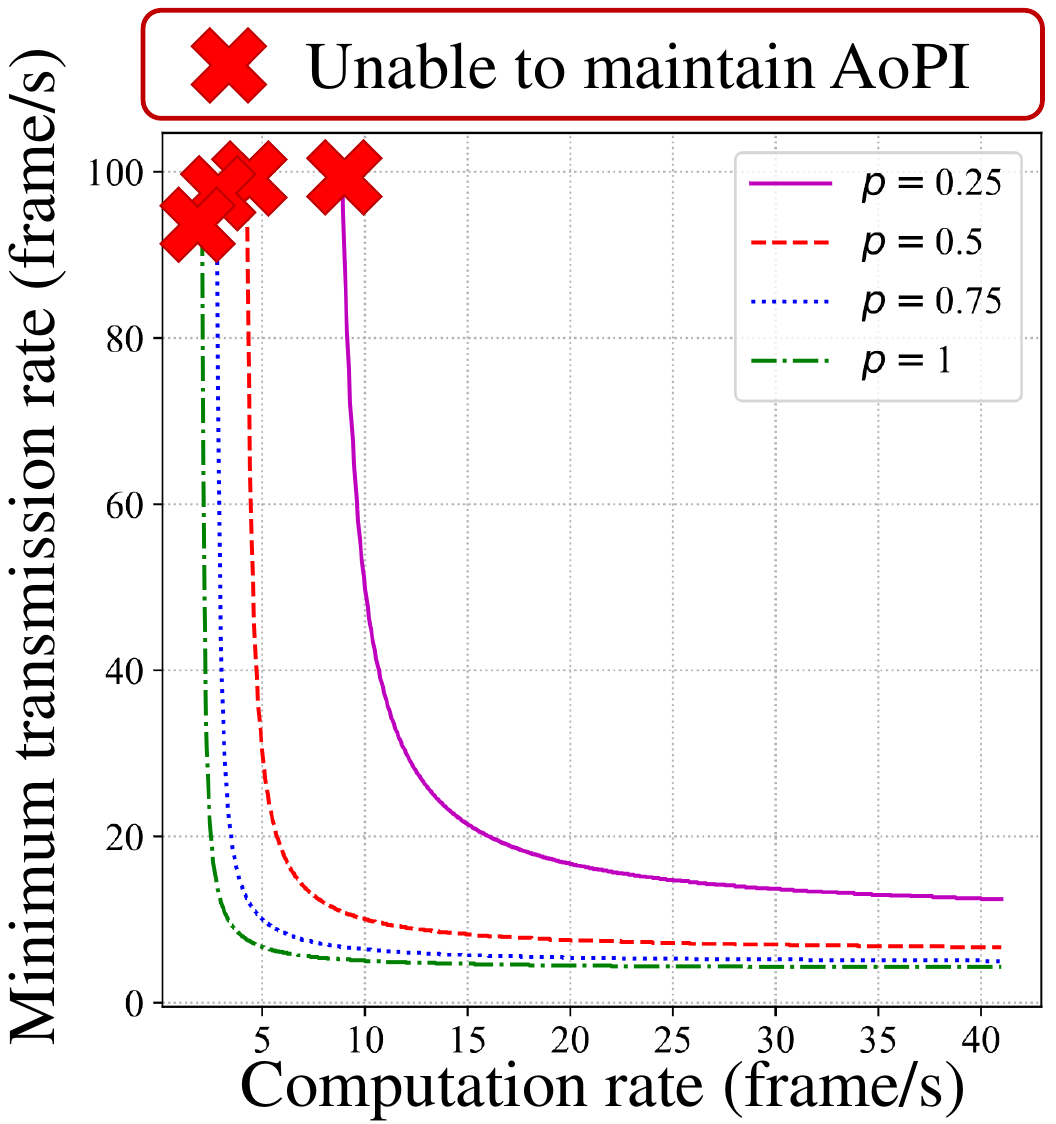}
        \caption{}
		\label{lcfsp_trade_t1}
	\end{subfigure}
	\centering
	\begin{subfigure}{0.48\linewidth}
		\centering
		\includegraphics[width=0.95\linewidth]{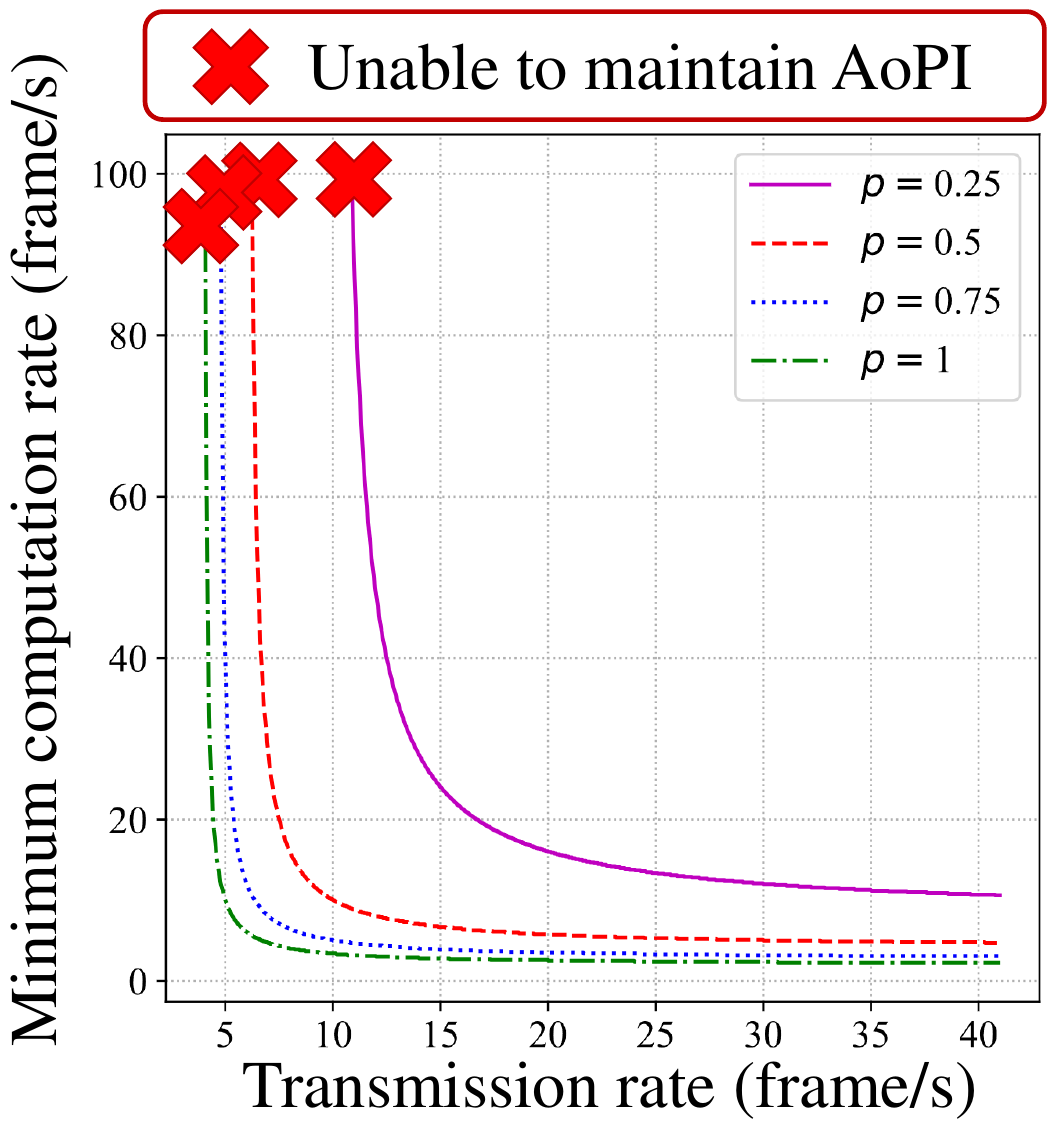}
        \caption{}
		\label{lcfsp_trade_c3}
	\end{subfigure}
	\caption{Minimum transmission rate (a) and computation rate (b) required to keep the average AoPI lower than 0.5s under the LCFSP policy.}
	\label{lcfsp_trade}
\end{figure}

\begin{figure}[t]
\centering
\includegraphics[width=0.65\linewidth]{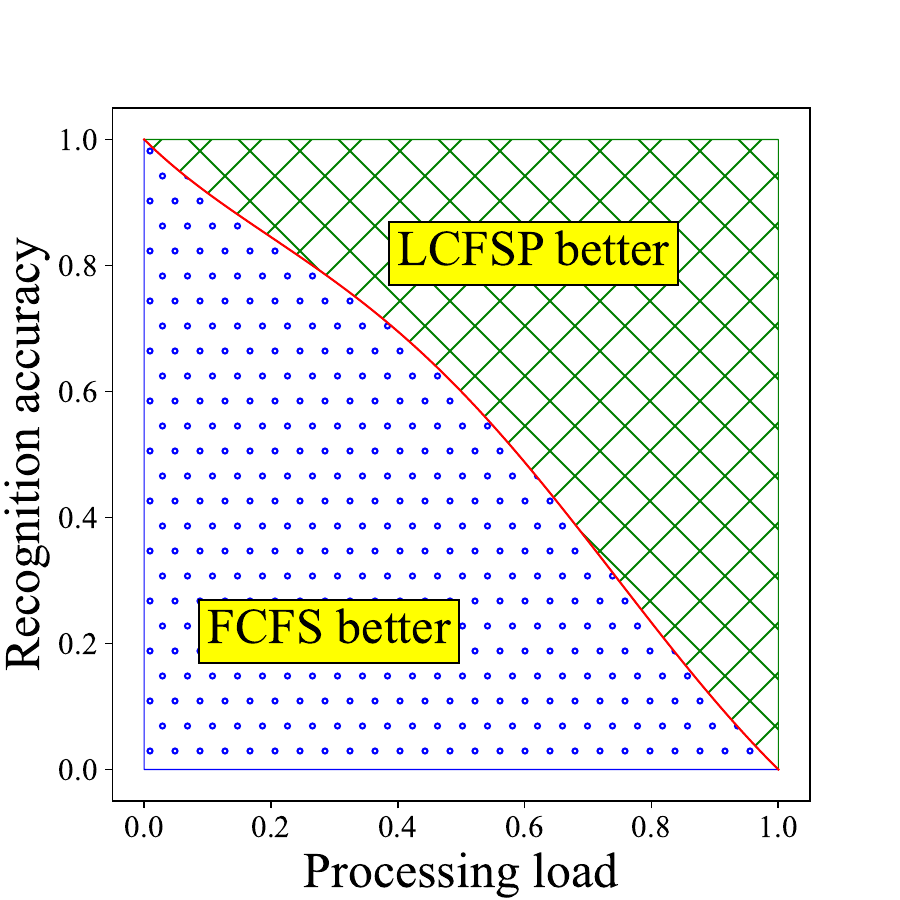}
\caption{The optimal policy with various recognition accuracy and processing load.}
\label{Fig.thre}
\end{figure}

\subsection{Comparison between Two Policies}

\begin{theorem}
For camera \(n\) in slot \(t\), the FCFS policy's AoPI exceeds the LCFSP policy if and only if the recognition accuracy \(p\) satisfies:
\begin{equation}
p \ge \frac{1-\rho ^2}{2\rho ^3-2\rho ^2+\rho +1}  ,
\end{equation}
where \(\rho=\lambda/\mu \) represents the processing load on edge server.
\end{theorem}
\begin{proof}
Proved by comparing Theorem 1 and 2.
\end{proof}

Based on Theorem 3, it is proved that the LCFSP policy performs better when the recognition accuracy \(p\) is higher than a certain threshold, i.e., in the cases of high processing load and recognition accuracy, as depicted in Fig. \ref{Fig.thre}. The insight is that high transmission rate results in large queuing delay under the FCFS policy. This issue can be mitigated by abandoning the older frames, which makes the LCFSP policy achieve better performance. On the other hand, in the cases of low recognition accuracy, frequent preemption under the LCFSP policy makes the accurate results harder to appear, while FCFS policy ensures no accurate results are missed.

We build a prototype to validate the above analysis results. According to Section \ref{pv}, the theoretical AoPI under the FCFS and LCFSP policies (computed based on Theorem 1 and 2) matches the experimental AoPI well. We find that the analytical results are effective for choosing the optimal video configuration in practical systems.

\section{Online Algorithm Design}
\label{oad}
In Section \ref{ave_aoi}, we derived the closed-form expressions of the AoPI of camera \(n\) in slot \(t\) under FCFS and LCFSP policy, which are the objective function in problem (P1). As a mixed-integer programming problem, (P1) is extremely challenging to solve due to its NP-hard nature and the lack of future information about the dynamics of video contents and network conditions.

In this section, we propose a light-weighted online method, i.e., Lyapunov-based Block Coordinate Descent (LBCD). LBCD first transforms the long-term optimization problem (P1) into a series of one-slot problems based on the Lyapunov framework (Section \ref{lyf}). LBCD makes real-time decisions by solving these independent one-slot problems, which require no future information.

\textcolor{black}{
LBCD further decouples each one-slot problem into two subproblems. In Section \ref{vca}, Subproblem 1 corresponds to adjusting video configurations and allocating transmission/computation resources given the edge server selection strategy of cameras. The proposed Algorithm 1 converges to the optimal solution to Subproblem 1 due to the convexity in the AoPI function. In Section \ref{ess}, the Subproblem 2 corresponds to assigning cameras to different edge servers, which is solved by adopting Algorithm 2 to match the resource requirements of cameras and capacity of edge servers. This decoupling approach enables the iteration of all decision variables towards the global optimum in the most efficient manner. Specifically, we prove that the LBCD method can achieve asymptotically optimal performance.}

\subsection{Lyapunov Framework}
\label{lyf}
The long-term recognition accuracy constraint (9) in problem (P1) couples the edge server selection, video configuration adaptation and resource allocation decisions across different slots. To address this challenge, we first define a virtual queue \(q\left ( t \right )\) as a historical measurement of the recognition accuracy overflow:
\begin{equation}
\label{qt}
    q\left ( t+1 \right ) = \max \left [q\left ( t \right ) -\bar{P} _t  +P_{\mathrm{min} },0 \right ] ,
\end{equation}
where \(\bar{P} _t  =\frac{1}{N} \sum_{n=1}^{N} p_{n,t}\) and \(q\left ( 0 \right ) =0\) for initialization.

If the cameras continuously choose resolutions with low accuracy, \(q\left ( t \right )\) will significantly increase, leading to poor user experience. Hence, we define a quadratic Lyapunov function to help keep \(q\left ( t \right )\) stable:
\begin{equation}
\label{lf}
    L\left ( q\left ( t \right )  \right ) =\frac{1}{2} \left ( q\left ( t \right )  \right ) ^2.
\end{equation}
Then, we define the one-slot Lyapunov drift as:
\begin{equation}
\label{ld}
    \Delta \left ( q\left ( t \right )  \right ) =\mathbb{E} \left [ L\left ( q\left ( t+1 \right )  \right ) 
-L\left ( q\left ( t \right )  \right )|q\left ( t \right ) \right ] ,
\end{equation}
which represents the expected change of the Lyapunov function in time slot \(t\). A small \(\Delta \left ( q\left ( t \right )  \right )\) means that the virtual queue \(q\left ( t \right )\) keeps stable. Besides, \(\Delta \left ( q\left ( t \right )  \right )\) satisfies the following lemma.

\begin{lemma}
For all possible values of \(q\left ( t \right )\) obtained by any feasible strategy, the following statement holds:
\begin{equation}
    \Delta \left ( q\left ( t \right )  \right )
\le \frac{1}{2} + q\left ( t \right ) \mathbb{E} \left [ P_{\mathrm{min} } -\bar{P} _t |q\left ( t \right )\right ] .
\end{equation}
\end{lemma}
\begin{proof}
    From Eq. (\ref{qt}) and (\ref{lf}), we have:
\begin{align}
\Delta(q(t))  &=\frac{1}{2} \mathbb{E}\left[(q(t+1))^{2}-(q(t))^{2} \mid q(t)\right] \nonumber \\
&\leq \frac{1}{2} \mathbb{E}\left[\left(q(t)-\bar{P} _t+P_{\mathrm{min} }\right)^{2}-(q(t))^{2} \mid q(t)\right] \nonumber \\
&=\frac{1}{2} \mathbb{E}\left[\left(P_{\mathrm{min} }-\bar{P} _t\right)^{2} \mid q(t)\right] \nonumber \\ &\quad+\mathbb{E}\left[q(t)\left(P_{\mathrm{min} }-\bar{P} _t\right) \mid q(t)\right] \nonumber \\
&\overset{(a)}{\leq}  \frac{1}{2}+q(t) \mathbb{E}\left[P_{\mathrm{min} }-\bar{P} _t \mid q(t)\right],
\end{align}
where step (a) is due to $P_{\mathrm{min} },\bar{P} _t\in\left [ 0,1 \right ] $.
\end{proof}

Next, we introduce a Lyapunov drift-plus-penalty function:
\begin{equation}
\label{ldpp}
    \Delta \left ( q\left ( t \right )  \right ) + V\cdot \mathbb{E} \left [ \bar{A} _t|q\left ( t \right )  \right ] ,
\end{equation}
where \(V\) is a positive parameter and $\bar{A} _t=\frac{1}{N} \sum_{n=1}^{N}A_{n,t}$. By optimizing the Lyapunov drift-plus-penalty function, we can balance the AoPI minimization objective and recognition accuracy constraint. However, it is extremely hard to directly optimize Eq. (\ref{ldpp}). Instead, we try to minimize its upper bound in each slot, which can be derived from Lemma 1:
\begin{align}
\label{lydftu}
&\Delta \left ( q\left ( t \right )  \right ) + V\cdot \mathbb{E} \left [ \bar{A} _t|q\left ( t \right )  \right ]
\nonumber \\ &\le \frac{1}{2} + q\left ( t \right ) \mathbb{E} \left [ P_{\mathrm{min} } -\bar{P} _t |q\left ( t \right )\right ] + V\cdot \mathbb{E} \left [ \bar{A} _t|q\left ( t \right )  \right ]. 
\end{align}
Finally, the new real-time optimization problem is:
\begin{align}
\mathbf{\left ( P2 \right )} & \min_{\left \{ \mathbf{y},\mathbf{r},\mathbf{x},\atop \mathbf{m}, \mathbf{b},\mathbf{c} \right \} } -q\left ( t \right ) \cdot \bar{P} _t +V\cdot \bar{A} _t
\\ & s.t. \ \left ( 6 \right ) \left ( 7 \right ) \left ( 8 \right ) \left ( 9 \right ) \left ( 10 \right ). \nonumber
\end{align}

In problem (P2), only the information of current slot $t$ is required. Besides, the additional term $-q\left ( t \right ) \cdot \bar{P} _t$ enables the trade-off between AoPI optimization and accuracy maintenance. Specifically, if previous slots choose to upload low-resolution videos, $q\left ( t \right )$ will increase, and enhancing accuracy becomes more critical.

As the original problem (P1) is transformed into a series of independent problems, we only need to solve problem (P2) in each slot. Nevertheless, problem (P2) is a mixed-integer non-linear programming problem, which is still NP-hard. Although some existing approaches, e.g., mathematical solver, can obtain near-optimal solution to problem (P2), these approaches cannot make decisions in real-time due to the extremely high time complexity. In this regard, we further decouple problem (P2) into two subproblems as follows.

\begin{algorithm}[t]
\caption{Adapting video configurations and allocating bandwidth and computation resources}
\label{lbcd}
\begin{algorithmic}[1]
\STATE Set an initial configuration selection and resource allocation strategy $\left \{ \mathbf{r}_t^0 ,\mathbf{x}_t^0, \mathbf{m}_t^0 ,\mathbf{b}_t^0 ,\mathbf{c}_t^0  \right \} $;
\FOR{$i=1$ to $M$}
\STATE {Given $\left \{ \mathbf{b}_t^{i-1} ,\mathbf{c}_t^{i-1}  \right \} $, compute the solution of problem (\ref{52}), $\left \{ \mathbf{r}_t^{i} ,\mathbf{x}_t^{i},\mathbf{m}_t^i   \right \} $, via exhaustive search;}
\STATE Given $\left \{ \mathbf{r}_t^{i} ,\mathbf{x}_t^{i}, \mathbf{m}_t^i ,\mathbf{c}_t^{i-1}  \right \} $, compute the solution of problem (\ref{53}), $\left \{ \mathbf{b}_t^{i}  \right \}, $ via interior point method;
\STATE Given $\left \{ \mathbf{r}_t^{i} ,\mathbf{x}_t^{i}, \mathbf{m}_t^i ,\mathbf{b}_t^{i}  \right \} $, compute the solution of problem (\ref{54}), $\left \{ \mathbf{c}_t^{i}  \right \} $, via interior point method;
\ENDFOR
\STATE Get the strategy $\left \{ \mathbf{r}_t^M ,\mathbf{x}_t^M ,\mathbf{r}_t^M ,\mathbf{b}_t^M ,\mathbf{c}_t^M  \right \}$;
\end{algorithmic}
\end{algorithm}

\subsection{Subproblem 1: Video Configuration Adaptation and Resource Allocation}
\label{vca}
In the first subproblem, the edge server selection strategy $y_{n,t}^s$ of cameras is given. We classify the remaining decision variables into three groups, i.e., video configuration $\left \{ \mathbf{r}_t ,\mathbf{x}_t,\mathbf{m}_t \right \} $, bandwidth allocation $\mathbf{b}_t=\left \{ b_{1,t},\cdots ,b_{n,t} \right \} $ and computation resource allocation $\mathbf{c}_t=\left \{ c_{1,t},\cdots ,c_{n,t} \right \} $. We have $\mathbf{r}_t=\left \{ r_{1,t},\cdots ,r_{n,t} \right \} $, $\mathbf{x}_t=\left \{ x_{1,t},\cdots ,x_{n,t} \right \} $ and $\mathbf{m}_t=\left \{ m_{1,t},\cdots ,m_{n,t} \right \} $. To tackle this subproblem, we employ the block coordinate descent method, which is renowned for its rapid convergence and near-optimal performance in non-linear programming domains. Specifically, we will fix two groups of variables and optimize the rest group. After all three groups are optimized, we start a new iteration of optimization. Such iteration ends when the maximum iteration number is reached. As summarized in Algorithm \ref{lbcd}, we perform the following steps during each iteration:
\begin{itemize}
    \item \textcolor{black}{\textbf{Optimizing video configuration decisions (Algorithm \ref{lbcd}, line 3):} given $\mathbf{b}_t $ and $\mathbf{c}_t $, we need to solve:
    \begin{align}
    \label{52}
    & \min_{\left \{ \mathbf{r},\mathbf{x},\mathbf{m} \right \} } -q\left ( t \right ) \cdot \bar{P} _t +V\cdot \bar{A} _t
    \\ & s.t. \ \left ( 6 \right ) \left ( 7 \right ) \left ( 8 \right ) \left ( 9 \right ) \left ( 10 \right ). \nonumber
    \end{align}}
    
    Note that the possible combinations of video resolution, computation policy and neural network model for each camera are very limited. This problem can be directly solved via exhaustive search.
    
    \item \textcolor{black}{\textbf{Optimizing bandwidth allocation variables (Algorithm \ref{lbcd}, line 4):} given $\mathbf{r}_t$, $\mathbf{x}_t$, $\mathbf{m}_t$ and $\mathbf{c}_t$, we need to solve:
    \begin{align}
    \label{53}
    & \min_{\left \{ \mathbf{b}\right \} } \ V\cdot \bar{A} _t
    \\ & s.t. \ \left ( 6 \right ) \left ( 7 \right ) \left ( 8 \right ) \left ( 9 \right ) \left ( 10 \right ). \nonumber
    \end{align}}

    Note that the bandwidth allocation strategy does not impact the recognition accuracy. According to Corollary 4.1 and Theorem 2, this problem is a constrained convex optimization problem. Although the closed-form solutions can hardly be derived due to $\bar{A} _t$ contains fraction of fourth degree polynomials, we can obtain the optimal solution via interior point methods.
    
    \item \textcolor{black}{\textbf{Optimizing computation resource allocation variables (Algorithm \ref{lbcd}, line 5):} given $\mathbf{r}_t$, $\mathbf{x}_t$, $\mathbf{m}_t$ and $\mathbf{b}_t$, we need to solve:
    \begin{align}
    \label{54}
    & \min_{\left \{ \mathbf{c} \right \} } \ V\cdot \bar{A} _t
    \\ & s.t. \ \left ( 6 \right ) \left ( 7 \right ) \left ( 8 \right ) \left ( 9 \right ) \left ( 10 \right ). \nonumber
    \end{align}}
    
    Similarly, the computation resource strategy does not impact the recognition accuracy. This problem is also a constrained convex optimization problem. The optimal solution can be obtained by interior point methods.
\end{itemize}

\textcolor{black}{
In the above steps, the problems (\ref{53}) and (\ref{54}) are constrained convex optimization problems. Besides, the problem (\ref{52}) can also be easily transformed into a convex optimization problem by modifying the video configuration variables to be continuous. Therefore, according to \cite{DBLP:journals/siamrev/BottouCN18}, the Algorithm 1 can converge to the global optimal solution to Subproblem 1.}

\begin{algorithm}[t]
\caption{Edge server selection}
\label{mlbcd}
\begin{algorithmic}[1]
\STATE {Assume all cameras connect with an virtual edge server whose capacity is the sum of all edge servers;}
\STATE {Obtain the virtual edge server's resource allocation decisions using Algorithm \ref{lbcd};}
\STATE {Set the size of each camera and the volume of each edge server according to Eq. (56) and (57), sort all cameras and edge servers;}
\FOR{each sorted camera $n$}
\STATE {Assign camera $n$ to the first edge server with sufficient remaining bandwidth and computation resources;}
\IF{All edge servers cannot meet camera $n$'s demand;}
\STATE Assign camera $n$ to the edge server with highest remaining resources;
\ENDIF
\ENDFOR
\STATE {Given the camera assignment strategy, re-compute the configuration selection and resource allocation decisions for all cameras using Algorithm \ref{lbcd};}
\end{algorithmic}
\end{algorithm}

\subsection{Subproblem 2: Edge Server Selection}
\label{ess}
The second subproblem focuses on how to efficiently assign cameras to different edge servers such that unnecessary resource contention can be minimized. We handle this challenge via taking the following steps in time each slot $t$:

\begin{itemize}   
    \item \textcolor{black}{\textbf{Obtaining the ideal resource allocation decisions (Algorithm \ref{mlbcd}, line 1-2):} This step gives the bandwidth/computation resource requirements of all cameras in the ideal environment. Consider a powerful virtual edge server $\vartheta $ whose capacity equals the sum of all real edge servers, i.e.,:
    \begin{align}
    B_t^\vartheta=\sum_{s=1}^{S} B_t^s,\
    C_t^\vartheta=\sum_{s=1}^{S} C_t^s.
    \end{align}
    Assume all cameras connect with the virtual edge server, we can obtain the corresponding configuration adaptation and resource allocation decisions $\left \{ \hat{\mathbf{r}}_t , \hat{\mathbf{x}}_t ,\hat{\mathbf{m}}_t ,\hat{\mathbf{b}}_t ,\hat{\mathbf{c}}_t  \right \} $ using Algorithm \ref{lbcd}.} 

     \item \textcolor{black}{\textbf{Assigning cameras to edge servers (Algorithm \ref{mlbcd}, line 3-9):} This step matches the resource requirements of cameras with the capacity of edge servers. This matching problem is a variant of 2D bin-packing problem, which is solved by the first-fit method. To balance the bandwidth and computation resources, we first define the \textit{size} $\phi _n$ of camera $n$ and the \textit{volume} $\psi _s$ of edge server $s$ as:
     \begin{align}
    &\phi _n=\frac{\hat{ b}_{n,t}}{\sum_{s=1}^{S} B_t^s}+\frac{\hat{ c}_{n,t}}{\sum_{s=1}^{S} C_t^s}\\
    &\psi _s=\frac{B_t^s}{\sum_{s=1}^{S} B_t^s}+\frac{C_t^s}{\sum_{s=1}^{S} B_t^s}.
    \end{align}
    Next, we sort all cameras and edge servers based on their size or volume. The sorted cameras are assigned to the first edge server with sufficient resources, i.e., first-fit method. If there does not exist such edge server, the cameras will choose the server with highest remaining resources.}

    \item \textcolor{black}{\textbf{Re-computing decisions for all cameras (Algorithm \ref{mlbcd}, line 10):} given the camera assignment strategy, we re-allocate bandwidth/computation resources on each edge server and select video configurations for all cameras using Algorithm \ref{lbcd}.}
\end{itemize}

This above steps can efficiently match resource-demanding cameras with high capacity edge servers and thus relieve the resource contention in each edge server.

\begin{algorithm}[t]
\caption{Lyapunov-based block coordinate descent (LBCD)}
\label{mlbcd1}
\begin{algorithmic}[1]
\FOR {Each time slot $t \in \mathcal{T} $}
\STATE {Obtain the available bandwidth and computation capacity of edge servers in current slot $t$;}
\STATE {Profile the recognition accuracy function $\zeta _n^t\left ( \cdot, \cdot  \right ) $ for each camera $n$;}
\STATE Solve problem (P2) using Algorithm 1 and 2 to obtain the edge server selection, video configuration adaptation,  bandwidth and computation resource allocation decisions;
\STATE Update $q\left ( t+1 \right ) $ based on Eq. (\ref{qt});
\ENDFOR
\end{algorithmic}
\end{algorithm}

\subsection{Analysis of LBCD}
The proposed LBCD method is summarized in Algorithm 3. At the beginning of each time slot, we first obtain the network conditions and profile the recognition accuracy of cameras. Then, we solve the corresponding one-slot problem (P2) to assign edge servers, adjust video configurations and allocate resources to cameras. Finally, we update the virtual queue $q\left ( t+1 \right ) $ according to Eq. (\ref{qt}) and move to the next slot.

We note that the time complexity of interior point method in Algorithm \ref{lbcd} is $O\left ( N^{3.5} \right ) $, where $N$ is the number of cameras. Hence, the time complexity of Algorithm \ref{lbcd} is $O\left ( MN^{3.5} \right ) $, where $M$ is the maximum number of optimization iterations. The time complexity of Algorithm 2 is $O\left ( MN^{3.5} + NS\right ) $, where $S$ is the number of edge servers. Generally, the proposed LBCD method consists of multiple executions of Algorithm 1 and 2. Thus, the time complexity of LBCD is $O\left ( MN^{3.5} + NS\right ) $.

\begin{theorem}
    LBCD achieves the following performance bounds in the long-term:
\begin{align}
\label{25}
&\frac{1}{T} \sum_{t = 1}^{T} \mathbb{E} \left [ \bar{A}_t  \right ] \le A_{\mathrm{opt} }+\frac{1}{V}\left ( \frac{1}{2} +\Phi _{\mathrm{max} } \right )  ,
\\& 
\label{pud}
\frac{1}{T} \sum_{t=1}^{T} \mathbb{E} \left [ \bar{P}_t  \right ] \ge P_{\mathrm{min} }-\frac{1}{\epsilon } \left ( \frac{1}{2} +V\cdot A_{\mathrm{max} } \right ) ,
\end{align}
where $A_{\mathrm{opt} }$ and $A_{\mathrm{max} }$ represent the optimal and the worst AoPI in problem (P1), respectively. Besides, $\Phi _{\mathrm{max} }$ represents the maximum gap between LBCD and the optimal objective function value in problem (P2). The positive constant $\epsilon $ represents the long-term recognition accuracy surplus achieved by a stationary control strategy.
\end{theorem}
\begin{proof}
According to \cite{DBLP:series/synthesis/2010Neely}, for any $\delta >0$, there exists a stationary and randomized policy $\Pi $ for problem (P1), which satisfies the following statement under the Lyapunov framework:
\begin{align}
\label{27}
\mathbb{E} \left [ A_t^{\Pi } \right ] \le A_t^*+\delta ,\\
\label{28}
\mathbb{E} \left [ P_{\mathrm{min} }-P_t^{\Pi } \right ] \le\delta,
\end{align}
where $A_t^*$, $A_t^{\Pi }$ and $P_t^{\Pi }$ represent the optimal AoPI in the one-slot problem (P2), the average AoPI and recognition accuracy of all cameras under the policy $\Pi $, respectively.

Next, the drift-plus-penalty function of LBCD satisfies:
\begin{align}
\Delta(q(t))+V \mathbb{E}\left[\bar{A}_t \mid q(t)\right] \leq \frac{1}{2}+\Phi_{\mathrm{max} }+\delta q(t)+V\left(A_t^*+\delta\right).
\end{align}

Sum the above results from slot $1$ to slot $T$, we have:
\begin{align}
&\sum_{t=1}^{T} \Delta(q(t))+V \sum_{t=1}^{T} \mathbb{E}\left[\bar{A}_t \mid q(t)\right] \nonumber
\\&=\mathbb{E}[L(q(T))-L(q(0))]+V \sum_{t  = 1}^{T} \mathbb{E}\left[\bar{A}_t\right] \nonumber
\\&\leq \left ( \frac{1}{2} +\Phi_{\mathrm{max} } \right )  T +\delta \sum_{t  = 1}^{T} q(t)+V \sum_{t  = 1}^{T}\left(A_t^*+\delta\right).
\end{align}

Given that $\mathbb{E}[L(q(T))-L(q(0))]>0$ and let $\delta \to 0$, we have:
\begin{equation}
\frac{1}{T} \sum_{t = 1}^{T} \mathbb{E} \left [ \bar{A}_t  \right ] \le A_{\mathrm{opt} }+\frac{1}{V}\left ( \frac{1}{2} +\Phi_{\mathrm{max} } \right ) ,
\end{equation}
where $A_{\mathrm{opt} }=\frac{1}{T} \sum_{t  = 1}^{T}A_t^*$. Eq. (\ref{25}) is thus proved.

Next, according to \cite{DBLP:series/synthesis/2010Neely}, there exists a constant $\epsilon >0$ and a policy $\Lambda $ for problem (P1) under the Lyapunov framework that satisfies:
\begin{align}
& \mathbb{E} \left [ A_t^{\Lambda } \right ] = \Psi \left ( \epsilon  \right ), \\
& \mathbb{E} \left [ P_{\mathrm{min} } - P_t^{\Lambda }\right ] \le -\epsilon .
\end{align}

Thus, the drift-plus-penalty function of LBCD satisfies:
\begin{align}
\label{dppel}
\Delta \left ( q\left ( t \right )  \right ) + V\cdot \mathbb{E} \left [ \bar{A} _t|q\left ( t \right )  \right ]
\le \frac{1}{2} -\epsilon  q\left ( t \right ) +V\Psi \left ( \epsilon  \right ).
\end{align}

Sum the above results from slot $1$ to slot $T$, we have:
\begin{align}
&\mathbb{E}[L(q(T))-L(q(0))]+V \sum_{t  = 1}^{T} \mathbb{E}\left[\bar{A} _t\right] \nonumber
\\&\leq\frac{T}{2} - \epsilon \sum_{t  = 1}^{T} q(t)+VT\Psi \left ( \epsilon  \right ).
\end{align}

Given that $\mathbb{E}[L(q(T))-L(q(0))]>0$ and let $\delta \to 0$ and $\Psi(\epsilon),\bar{A} _t\in \left [ 0,A_{\mathrm{max}} \right ] $, we have:
\begin{align}
\lim _{T \rightarrow \infty} \frac{1}{T} \sum_{t=1}^{T} q(t) &\leq \frac{1}{2\epsilon}+\frac{V}{\epsilon}\left[\Psi(\epsilon)-\frac{1}{T} \sum_{t=1}^{T} \mathbb{E}\left[\bar{A} _t\right]\right]
\nonumber \\& 
\label{36}
\leq \frac{1}{2\epsilon}+\frac{V}{\epsilon}A_{\mathrm{max}}.
\end{align}

Recall that:
\begin{align}
\lim _{T \rightarrow \infty} \frac{1}{T} \sum_{t=1}^{T} q(t) &=\lim _{T \rightarrow \infty} \frac{1}{T} \sum_{t=1}^{T} \max \left [q\left ( t \right ) -\bar{P} _t  +P_{\mathrm{min} },0 \right ] 
\nonumber\\&
\label{37}
\ge P_{\min }-\lim _{T \rightarrow \infty} \frac{1}{T} \sum_{t=1}^{T}\bar{P} _{t-1}.
\end{align}

Combine Eq. (\ref{36}) with (\ref{37}), we have:
\begin{equation}
\frac{1}{T} \sum_{t=1}^{T} \mathbb{E} \left [ \bar{P}_t  \right ] \ge P_{\mathrm{min} }-\frac{1}{\epsilon } \left ( \frac{1}{2} +V\cdot A_{\mathrm{max} } \right ) .
\end{equation}

Theorem 4 is thus proved.

\end{proof}

By Theorem 4, an asymptotically optimal AoPI can be achieved by using a large value of $V$, whereas the recognition accuracy may degrade. However, we find that such trade-off is not obvious since AoPI and accuracy are not entirely conflicting metrics, according to Section \ref{ihp}.

\section{Evaluation and Discussions}
\label{sim_res}
In this section, we carry out extensive simulations and testbed experiments to validate the accuracy of the analytical results and evaluate the performance of proposed algorithm.

\subsection{Evaluation Setup}
In following sections, we utilize the real-trace dataset of
bandwidth and computation capacity variations collected from Ghent city \cite{DBLP:journals/icl/HooftPWHRBT16} and Bitbrains datacenter \cite{DBLP:conf/ccgrid/ShenBI15}, respectively. Unless otherwise stated, the average wireless bandwidth and computation capacity of edge servers are 30 MHz and 50 TFLOPS, respectively. The available video resolution include 384p, 512p, 640p, 768p, 896p and 1024p. \textcolor{black}{The cameras are divided into three groups evenly. The first group accounts for the object recognition task, where the model candidates include YOLOv5n,YOLOv5s, YOLOv5m, YOLOv5l and YOLOv5x \cite{yolov5m}. The second group accounts for the semantic segmentation task, where the model candidates include FPN\cite{DBLP:conf/cvpr/LinDGHHB17} and U-Net\cite{DBLP:journals/corr/RonnebergerFB15}. The third group accounts for the instance segmentation task, where the model candidates include YOLACT\cite{DBLP:conf/iccv/BolyaZXL19} and Mask R-CNN\cite{DBLP:journals/corr/HeGDG17}.} The Cityscapes dataset \cite{DBLP:conf/cvpr/CordtsORREBFRS16} is utilized to profile the time-varying recognition accuracy functions w.r.t. video resolution and neural network model. We adopt the following baselines from the state-of-the-art video analytics studies.

\begin{itemize}
\item \textbf{DOS}: proposed by \cite{DBLP:journals/ton/RongWLWLH22} to select video resolution, allocate bandwidth and computation resources, aiming to maximize accuracy minus latency.
\item \textbf{JCAB}: proposed by \cite{DBLP:journals/ton/ZhangWJWQXL22} to select video resolution and allocate bandwidth, aiming to maximize accuracy under latency constraint. We extend this method by allocating computation resources proportional to the computational complexity of frames, according to \cite{DBLP:journals/twc/RenYCH18}.
\item \textbf{MIN}: lower bound of the average AoPI in problem (P1), which is achieved by ignoring the long-term recognition accuracy requirement and treating all edge servers as one entity.
\end{itemize}

We notice that existing studies usually focus on video resolution adaptation. Hence, for fair comparison, the computation policy and neural network model in \textbf{DOS} and \textbf{JCAB} methods are selected according to Theorem 3, given other video resolution and resource allocation decisions. Besides, the DOS method is mainly designed to allocate resources on a single server. In following discussions, we let DOS and LBCD share the same edge server selection strategy.

\begin{figure}[tbp]
	\centering
	\begin{subfigure}{0.48\linewidth}
		\centering
		\includegraphics[width=0.95\linewidth]{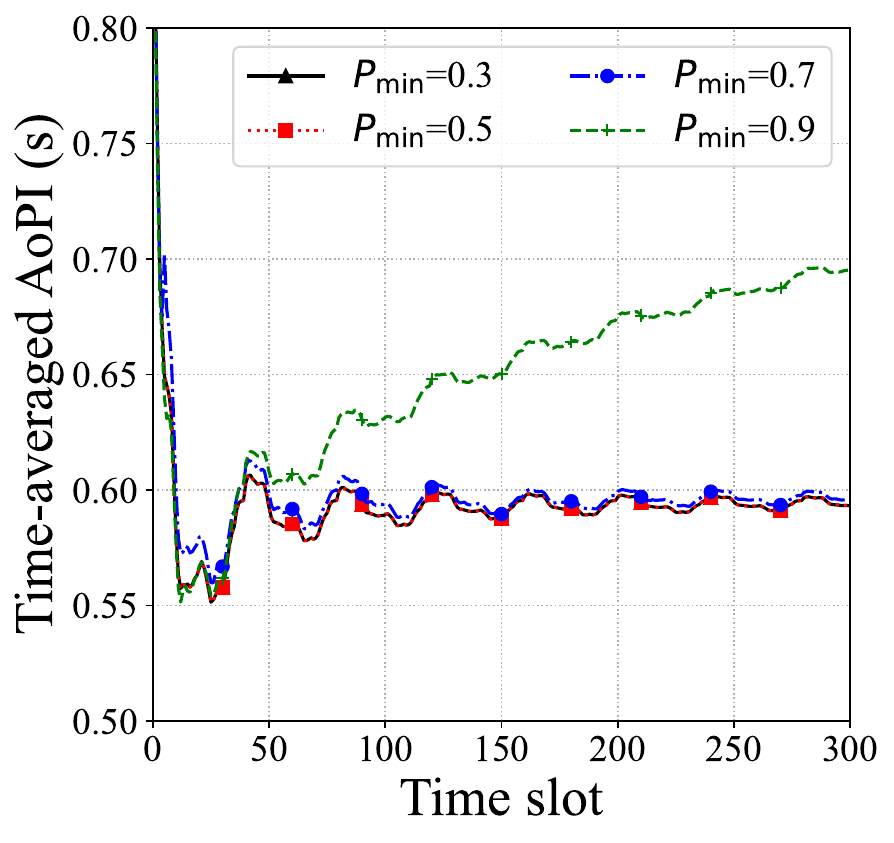}
        \caption{}
		\label{p_impact_a}
	\end{subfigure}
	\centering
	\begin{subfigure}{0.48\linewidth}
		\centering
		\includegraphics[width=0.95\linewidth]{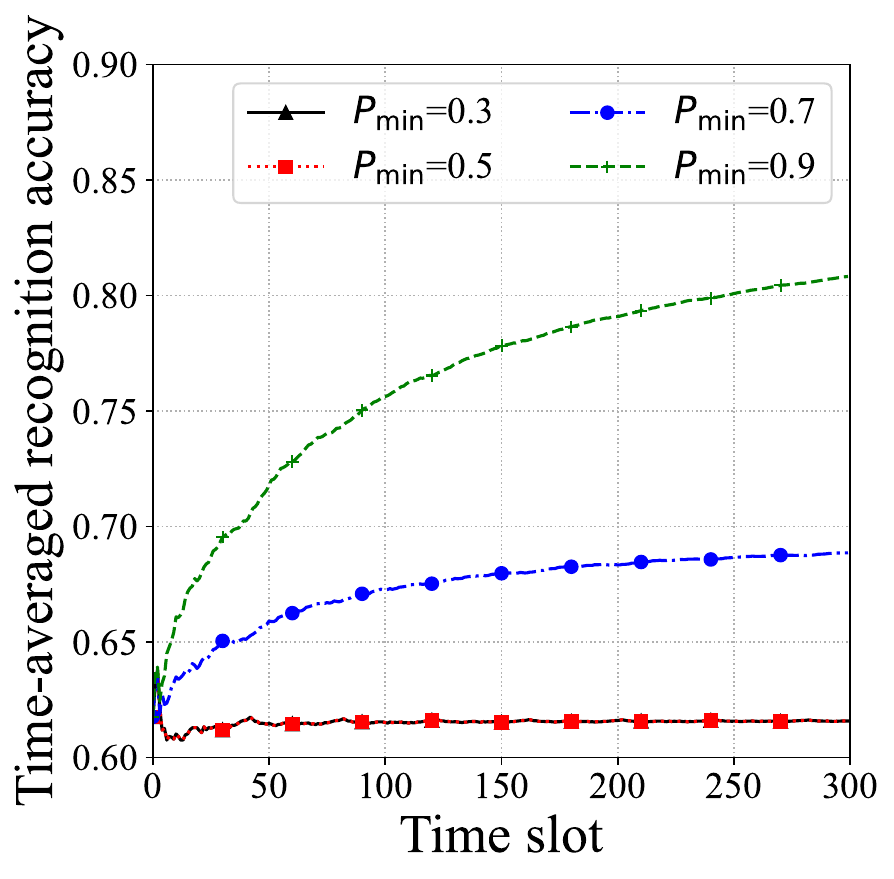}
        \caption{}
		\label{p_impact_b}
	\end{subfigure}
	\caption{Impact of the recognition accuracy threshold $P_{\mathrm{min} }$ on (a) average AoPI and (b) average recognition accuracy.}
	\label{p_impact}
\end{figure}

\begin{figure}[tbp]
	\centering
	\begin{subfigure}{0.48\linewidth}
		\centering
		\includegraphics[width=0.95\linewidth]{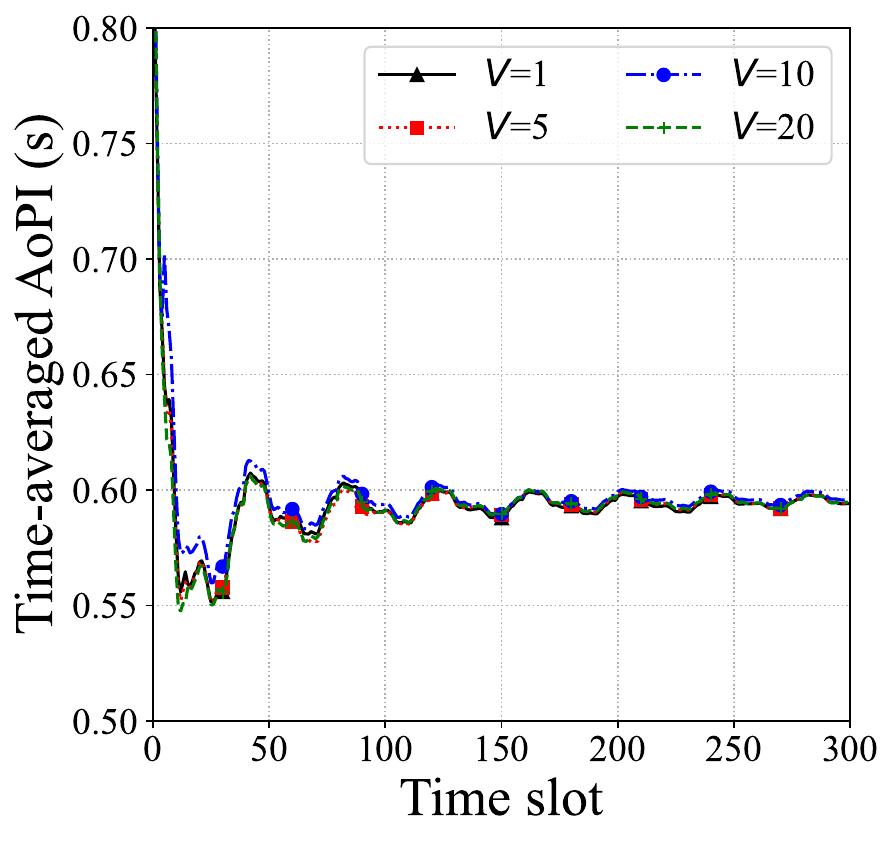}
        \caption{}
		\label{v_impact_a}
	\end{subfigure}
	\centering
	\begin{subfigure}{0.48\linewidth}
		\centering
		\includegraphics[width=0.95\linewidth]{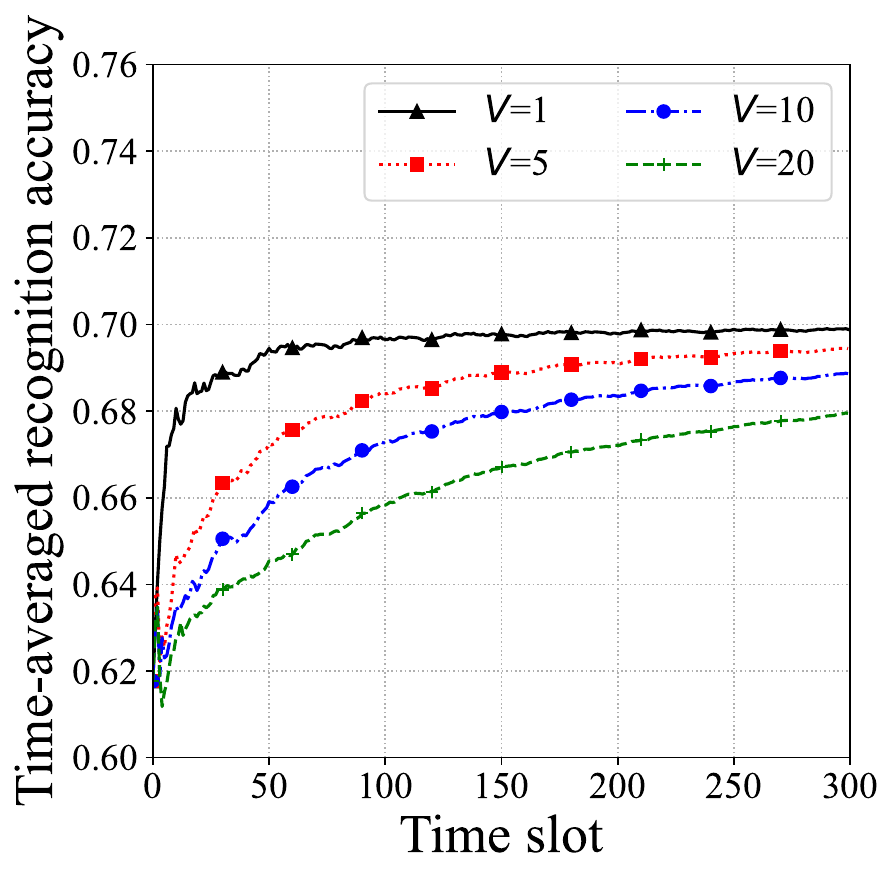}
        \caption{}
		\label{v_impact_b}
	\end{subfigure}
	\caption{Impact of the hyper-parameter $V$ on (a) average AoPI and (b) average recognition accuracy.}
	\label{v_impact}
\end{figure}

\subsection{Numerical Simulations}
In this section, we consider a system consisting of 3 edge servers and 30 cameras.

\subsubsection{Impact of Hyper-parameters}
\label{ihp}
\textcolor{black}{
We first analyze the impact of the long-term recognition accuracy threshold $P_{\mathrm{min} }$. According to Fig. \ref{p_impact}(a), the average AoPI surges when $P_{\mathrm{min} }$ increases to a very high level (e.g., 0.9). The main reason is that maintaining accuracy becomes the primary goal in such cases, which leads to a higher AoPI by always choosing high video resolution and large neural network model. From Fig. \ref{p_impact}(b), LBCD needs much more time slots to converge with a large $P_{\mathrm{min} }$. Interestingly, the recognition accuracy of LBCD always remains higher than 0.6 when $P_{\mathrm{min}}\le 0.5$. The main reason is that the video configuration with lowest AoPI has an accuracy of 0.61 on average. The insight is that AoPI and accuracy are not absolutely conflicting metrics, improving accuracy does not necessarily sacrifice AoPI. On the contrary, higher accuracy always enlarges latency in existing studies. In following discussions, we set $P_{\mathrm{min} }=0.7$, i.e., over 70\% video frames are accurately recognized, which can meet the requirements of many practical applications.}

The impact of the hyper-parameter $V$ in LBCD is depicted in Fig. \ref{v_impact}. As can be seen, LBCD converges much faster on the recognition accuracy with a small $V$. The main reason is that maintaining accuracy becomes the primary goal in such cases. With a large $V$, LBCD takes more time slots to meet the long-term recognition accuracy requirement, while achieving a slightly better AoPI. Thus, we set a median value $V=10$ in following discussions.

\begin{figure}[tbp]
	\centering
	\begin{subfigure}{0.48\linewidth}
		\centering
		\includegraphics[width=0.95\linewidth]{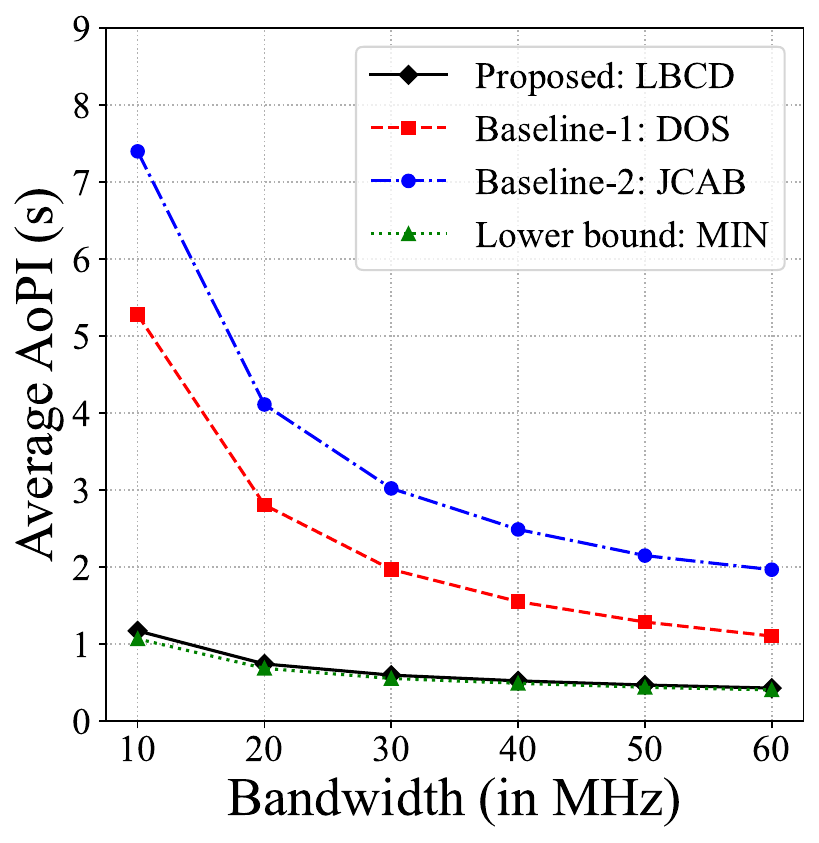}
        \caption{}
		\label{b_impact_a}
	\end{subfigure}
	\centering
	\begin{subfigure}{0.495\linewidth}
		\centering
		\includegraphics[width=0.95\linewidth]{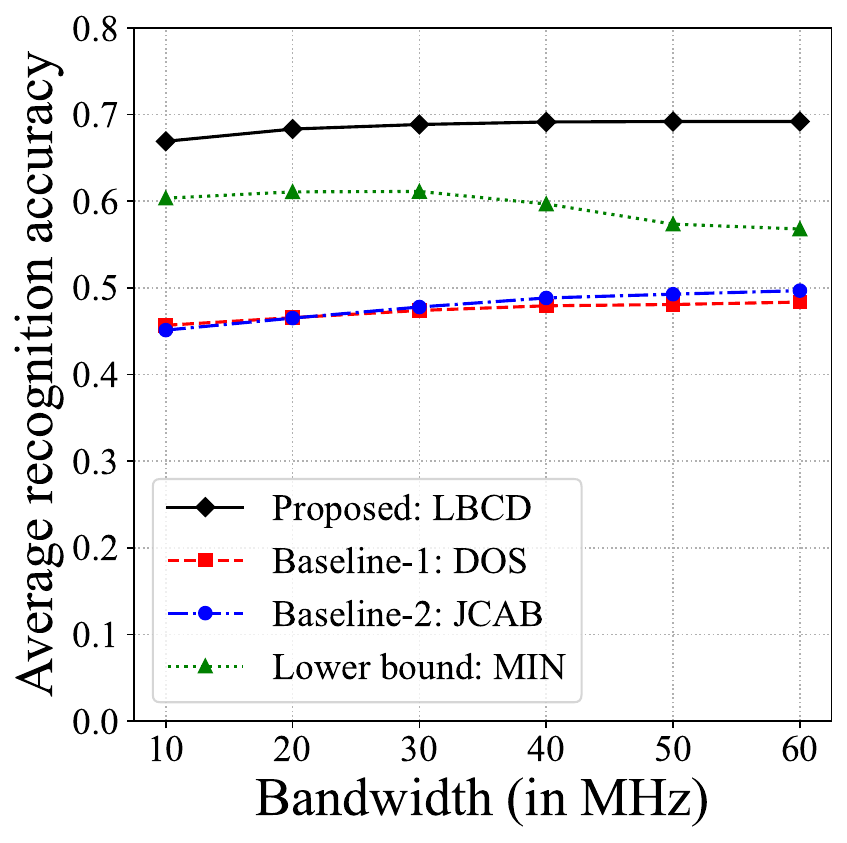}
        \caption{}
		\label{b_impact_b}
	\end{subfigure}
	\caption{Impact of bandwidth on (a) average AoPI and (b) average recognition accuracy.}
	\label{b_impact}
\end{figure}

\begin{figure}[tbp]
	\centering
	\begin{subfigure}{0.48\linewidth}
		\centering
		\includegraphics[width=0.95\linewidth]{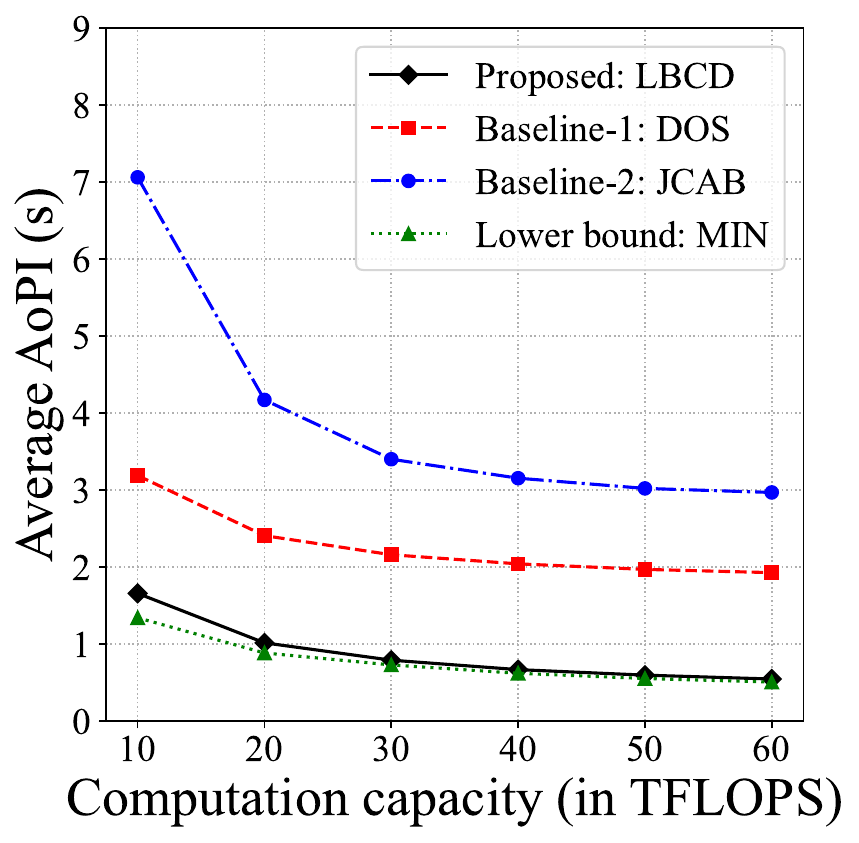}
        \caption{}
		\label{c_impact_a}
	\end{subfigure}
	\centering
	\begin{subfigure}{0.495\linewidth}
		\centering
		\includegraphics[width=0.95\linewidth]{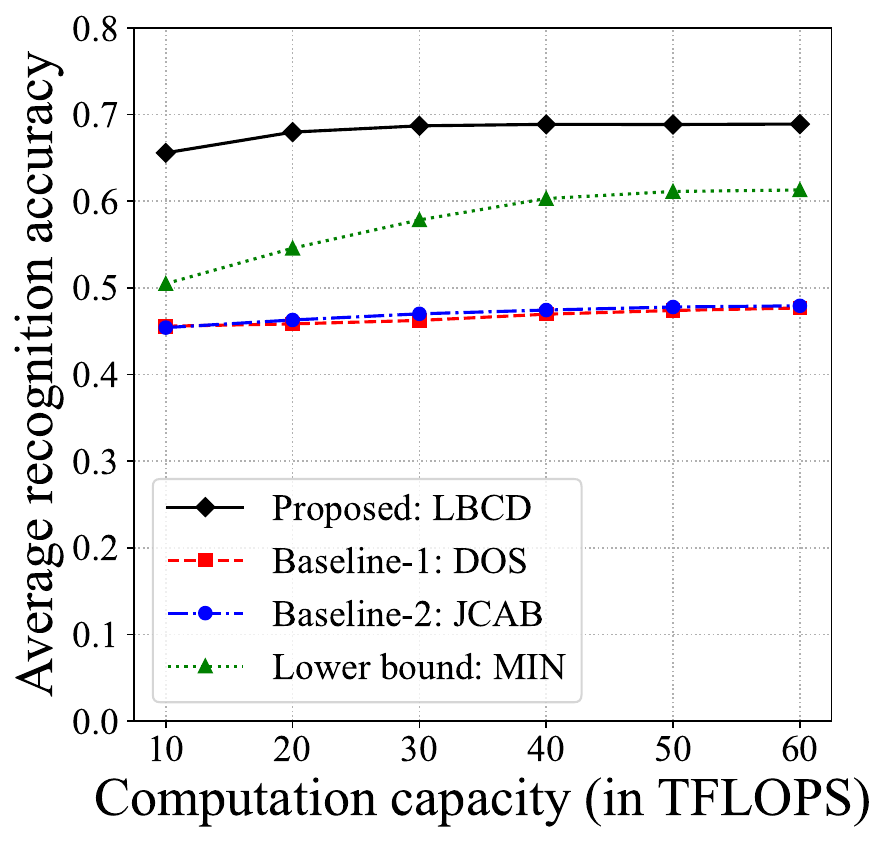}
        \caption{}
		\label{c_impact_b}
	\end{subfigure}
	\caption{Impact of computation capacity on (a) average AoPI and (b) average recognition accuracy.}
	\label{c_impact}
\end{figure}

\begin{figure}[t]
	\centering
	\begin{subfigure}{0.48\linewidth}
		\centering
		\includegraphics[width=0.95\linewidth]{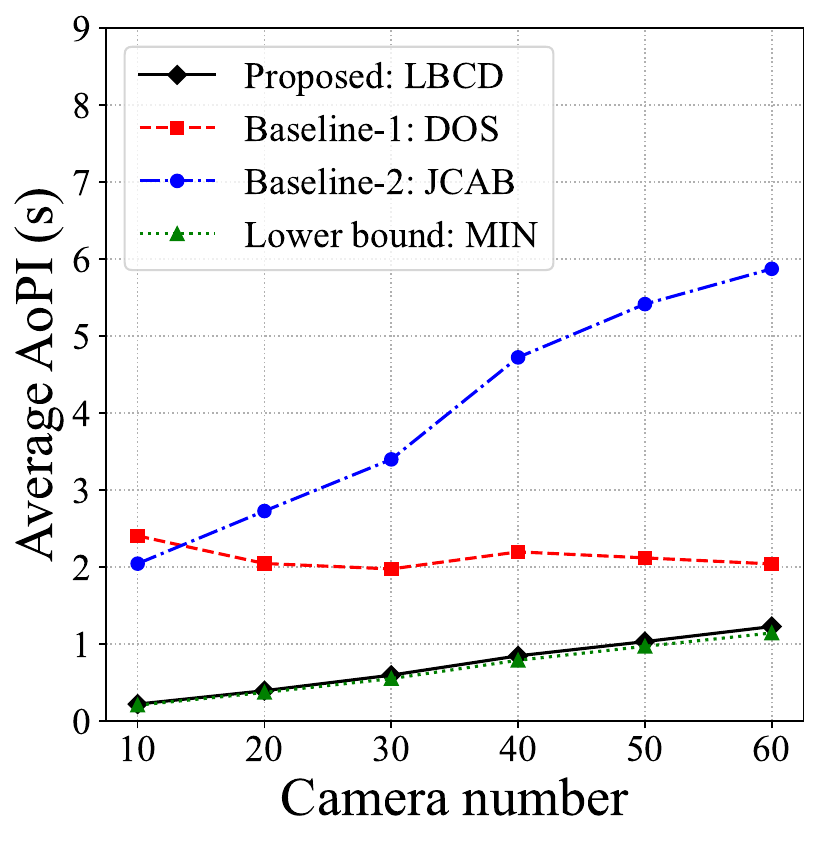}
        \caption{}
		\label{n_impact_a}
	\end{subfigure}
	\centering
	\begin{subfigure}{0.495\linewidth}
		\centering
		\includegraphics[width=0.95\linewidth]{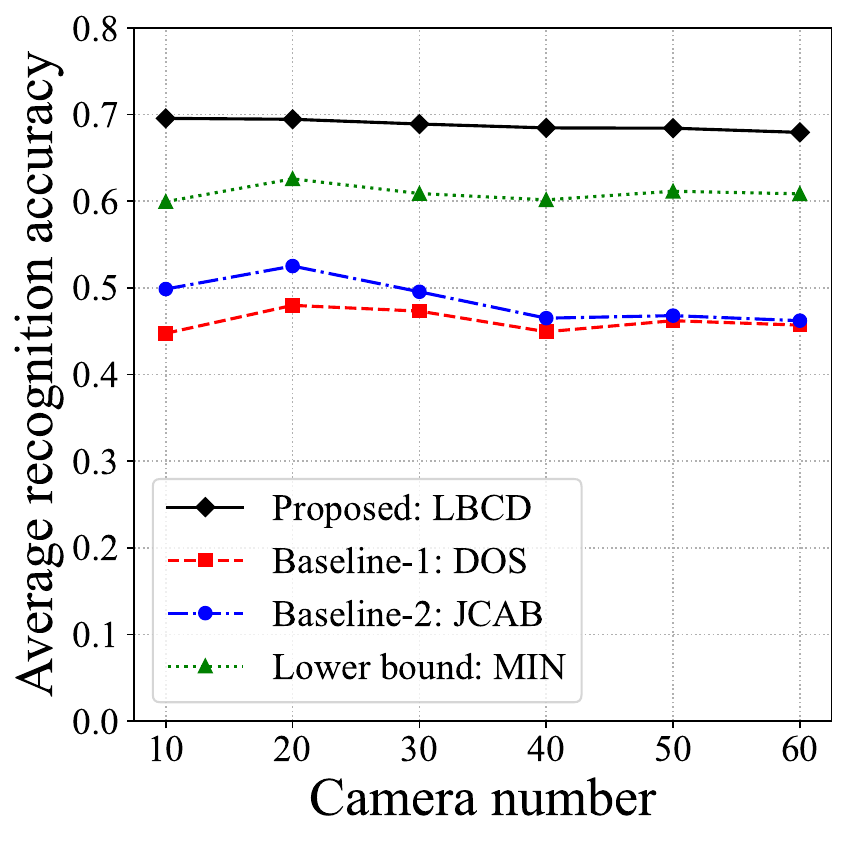}
        \caption{}
		\label{n_impact_b}
	\end{subfigure}
	\caption{Impact of camera number on (a) average AoPI and (b) average recognition accuracy.}
	\label{n_impact}
\end{figure}

\begin{figure}[t]
	\centering
	\begin{subfigure}{0.48\linewidth}
		\centering
		\includegraphics[width=0.95\linewidth]{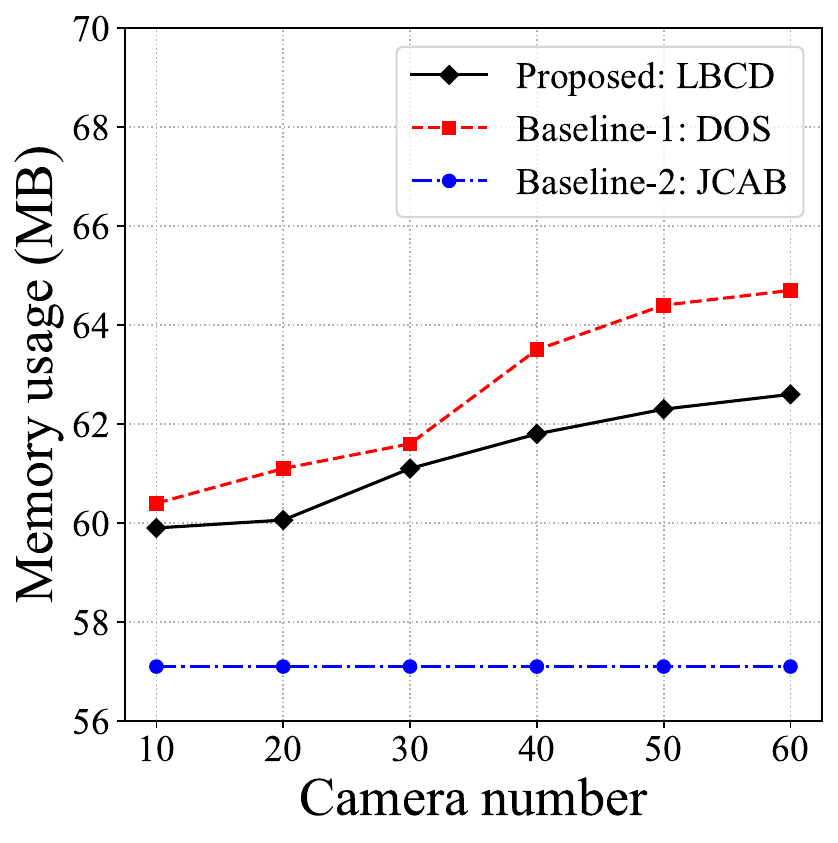}
        \caption{}
		\label{memory}
	\end{subfigure}
	\centering
	\begin{subfigure}{0.49\linewidth}
		\centering
		\includegraphics[width=0.95\linewidth]{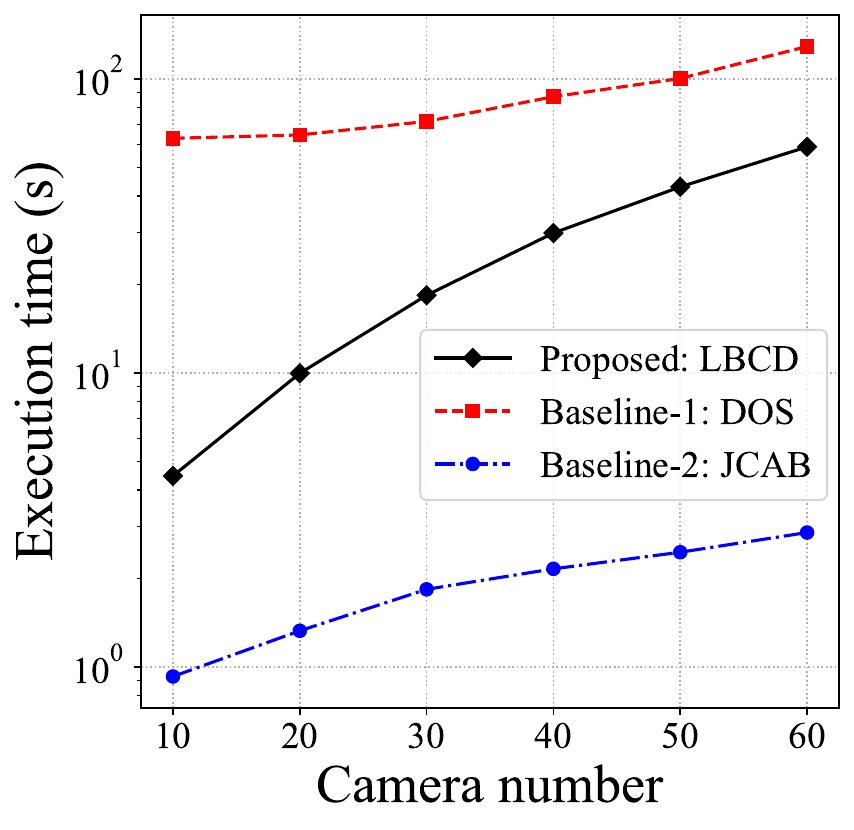}
        \caption{}
		\label{time}
	\end{subfigure}
	\caption{Memory usage (a) and execution time (b) of different methods.}
	\label{overhead}
\end{figure}

\subsubsection{Performance Comparison}
\label{perc}
\textcolor{black}{
The impact of wireless bandwidth is depicted in Fig. \ref{b_impact}. As can be seen, LBCD significantly outperforms other baselines and achieves close-to-optimal AoPI. Specifically, LBCD reduces AoPI by up to 6.32X compared with the JCAB method (with 20MHz bandwidth) and 4.51X with the DOS method (with 10MHz bandwidth).} According to Fig. \ref{b_impact}(a), all methods can achieve lower AoPI when the wireless bandwidth of edge servers increases. The main reason is that the transmission latency of video frames becomes low.

\textcolor{black}{
In Fig. \ref{b_impact}(b), LBCD achieves the highest recognition accuracy. Generally, higher wireless bandwidth helps LBCD to reach better recognition accuracy. The main reason is that LBCD has to choose low-resolution frames when the bandwidth is limited (e.g., 10MHz).} Note that the recognition accuracy of LBCD becomes stable when the bandwidth exceeds 20MHz. The main reason is that choosing video configurations with accuracy higher than the threshold $P_{\mathrm{min} }=0.7$ leads to worse AoPI in such cases. Interestingly, the accuracy of the MIN method decreases with more bandwidth. The main reason is that MIN will reduce the video resolution when the bandwidth increases without the requirement of recognition accuracy. This strategy leads to better AoPI because the significant reduction in latency can compensate for the slight drop in recognition accuracy. 

The JCAB and DOS method has the worst recognition accuracy in Fig. \ref{b_impact}(b). According to \cite{DBLP:journals/ton/ZhangWJWQXL22}, JCAB has to satisfy a constraint on the total transmission and computation latency when choosing video configurations\footnote{We follow the settings in \cite{DBLP:journals/ton/ZhangWJWQXL22} and set the latency constraint as 0.5s.}. JCAB will directly choose the highest video resolution and largest neural network model which can meet the latency constraint. Thus, a stringent constraint will lead to low recognition accuracy. Although a loose constraint helps increases accuracy, it leads to worse AoPI due to the high transmission/ computation latency.

According to \cite{DBLP:journals/ton/RongWLWLH22}, the DOS method minimizes the AoPI minus recognition accuracy. We note that the latency of different video configurations grows much faster than the accuracy. For example, when the resolution increases from 256p to 1024p (using the YOLOv5m model), the total latency grows from 0.02 to 0.6, while the total latency grows from 0.05s to 0.91s. Therefore, DOS always selects the lowest video resolution and lightest neural network model in our simulations, leading to poor accuracy. Although setting different weights helps DOS value accuracy more, it is impractical to use a fixed weight factor to accommodate all network conditions.

\textcolor{black}{
Next, the impact of computation capacity is shown in Fig. \ref{c_impact}. According to Fig. \ref{c_impact}(a), LBCD still achieves close-to-optimal AoPI. Specifically, LBCD reduces AoPI by up to 5.44X compared with the JCAB method and 3.53X with the DOS method when the computation capacity equals 60 TFLOPS. Generally, higher computation capacity helps all methods achieve lower AoPI due to the reduced computation latency.}

In Fig. \ref{c_impact}(b), higher computation capacity also helps to achieve higher recognition accuracy. Specifically, the LBCD method choose larger models with increased computation capacity, which has better accuracy without sacrificing the computation latency of frames. When the computation capacity is extremely limited, the MIN method has to choose a low resolution and small model to balance frames' transmission and computation latency. Then, with higher computation capacity, the accuracy of the MIN method increases rapidly because it switches to high resolution and large models. Finally, as discussed in Fig. \ref{b_impact}, the JCAB and DOS methods still have the worst accuracy.

\textcolor{black}{
The impact of the camera number is depicted in Fig. \ref{n_impact}. According to Fig. \ref{n_impact}(a), LBCD still achieves close-to-optimal AoPI. Specifically, LBCD reduces AoPI by up to 9.3X compared with the JCAB method and 10.94X with the DOS method (with 10 cameras). Besides, the AoPI of LBCD and JCAB methods increase near-linearly with more cameras due to the strengthened resource contention among all cameras. Interestingly, the AoPI of DOS method is relatively stable. The main reason is that its resource allocation strategy is much unbalanced. Correspondingly, the AoPI of most cameras are always high. From Fig. \ref{n_impact}(b), the recognition accuracy of all methods generally decreases with the camera number. However, LBCD still achieves a relatively stable and high recognition accuracy because it chooses the same video configuration when the allocated resources shrink linearly with the camera number.}

\textcolor{black}{
Finally, the computational overhead of different methods are shown in Fig. \ref{overhead}. As can be seen, all methods have similar memory usage of about 60 MB, which slightly increases with the cameras number. On the other hand, all methods have significant difference in terms of execution time. The LBCD method performs worse than JCAB but much better than DOS. Nevertheless, LBCD can still handle the case of 20 cameras within 10 seconds. In the practical environment, it is less likely to deploy more than 20 cameras in a small area, indicating that the execution time of LBCD is acceptable.}

\begin{figure}[t] 
\centering 
\includegraphics[width=\linewidth]{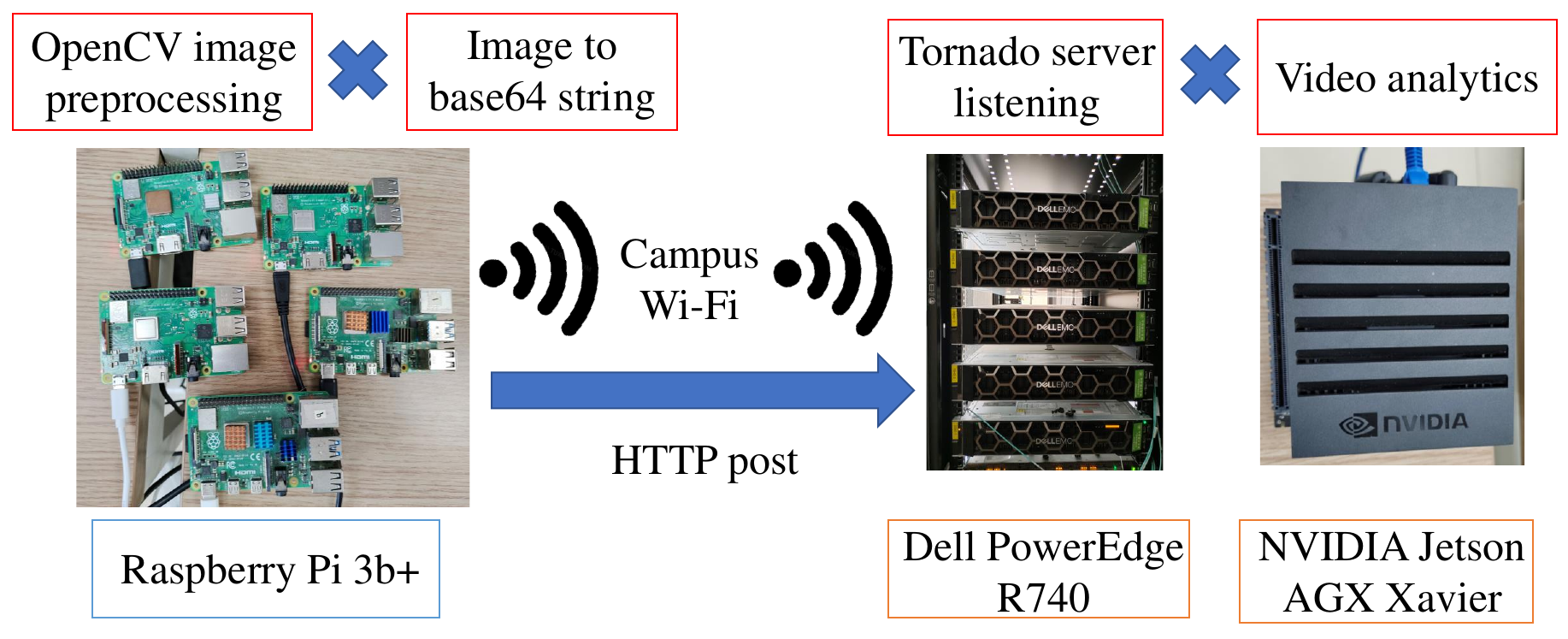} 
\caption{An illustration of the testbed.}
\label{pro}
\end{figure}

\begin{figure}[t]
	\centering
	\begin{subfigure}{0.48\linewidth}
		\centering
		\includegraphics[width=0.95\linewidth]{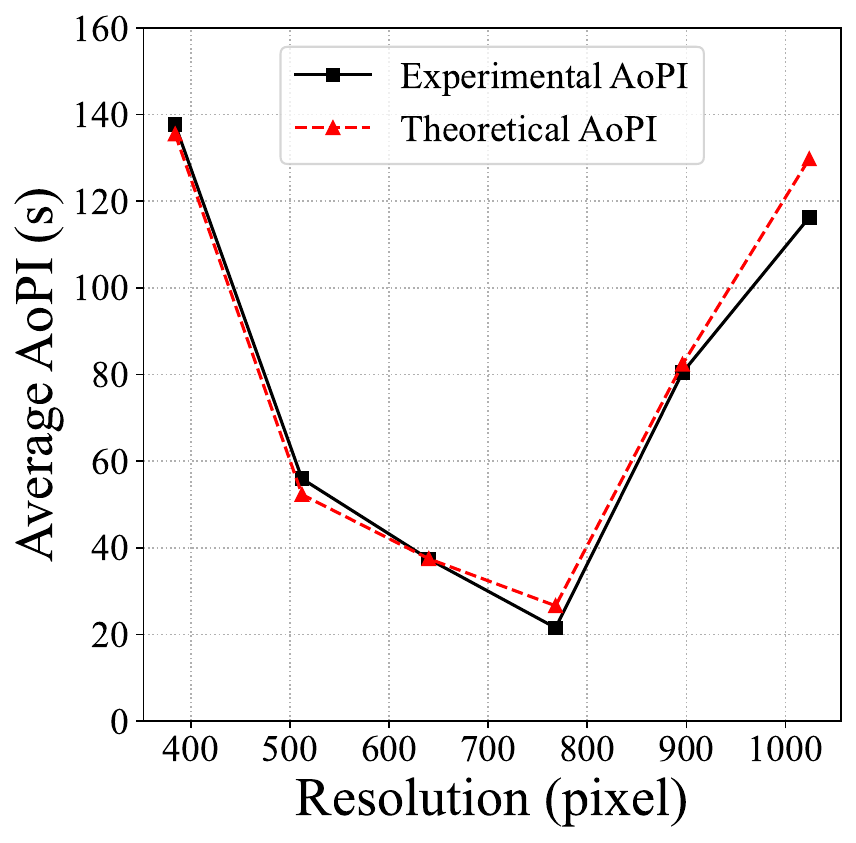}
        \caption{}
		\label{lcfsp_trade_t2}
	\end{subfigure}
	\centering
	\begin{subfigure}{0.48\linewidth}
		\centering
		\includegraphics[width=0.95\linewidth]{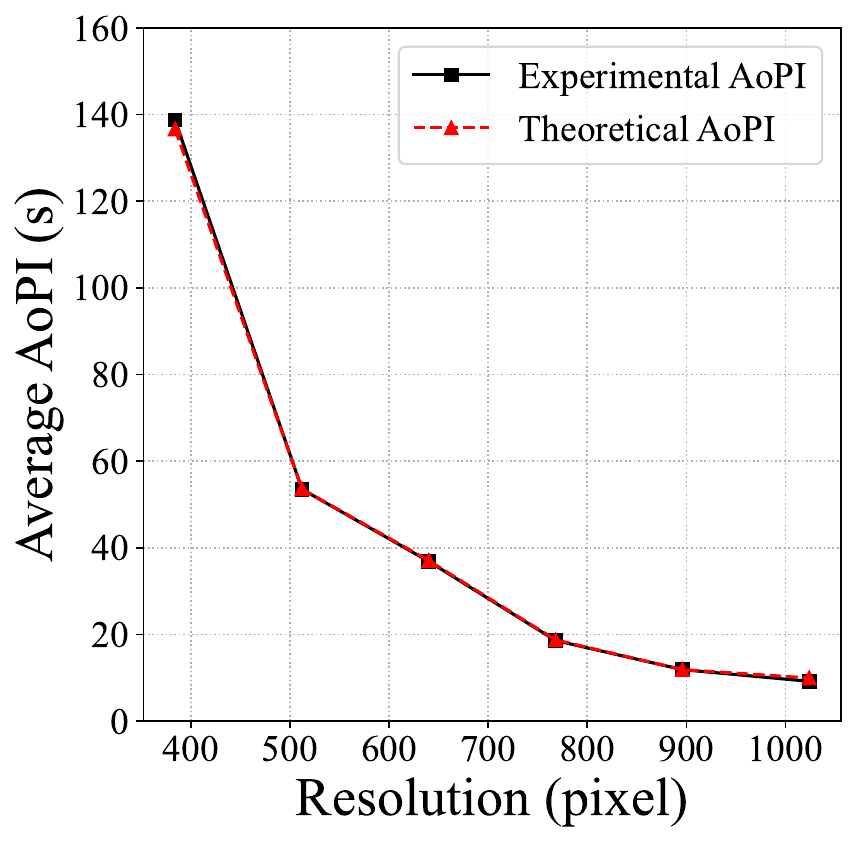}
        \caption{}
		\label{lcfsp_trade_c1}
	\end{subfigure}
	\caption{Comparison between the experimental and theoretical average AoPI on the CPU-based edge server under (a) FCFS and (b) LCFSP policy.}
	\label{et_age}
\end{figure}

\begin{figure}[t]
	\centering
	\begin{subfigure}{0.48\linewidth}
		\centering
		\includegraphics[width=0.95\linewidth]{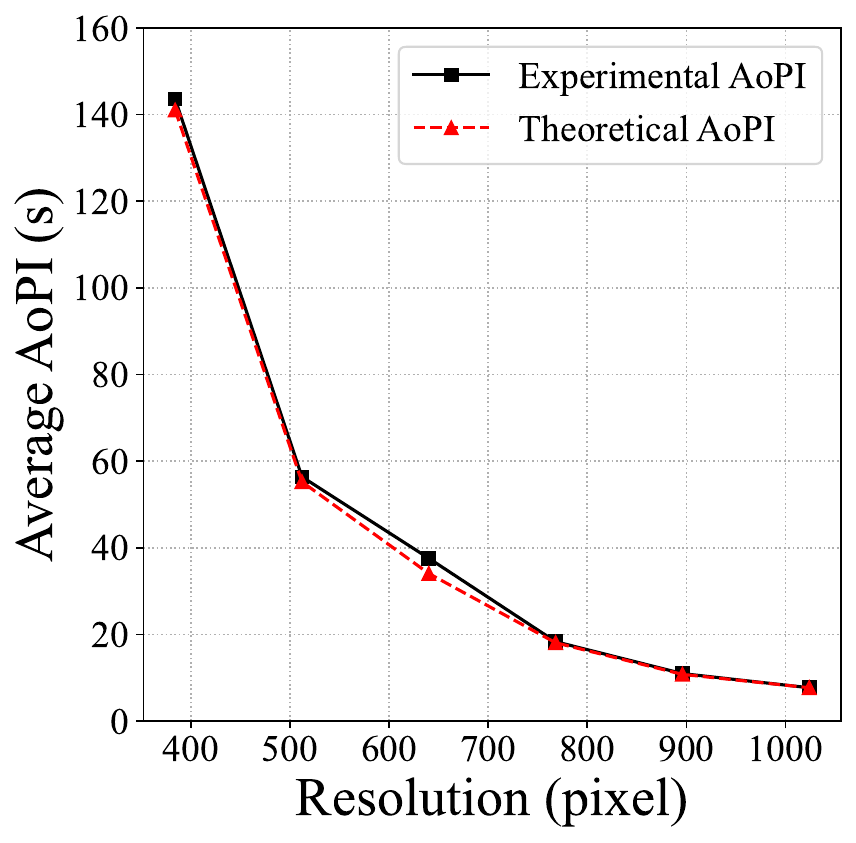}
        \caption{}
		\label{lcfsp_trade_t3}
	\end{subfigure}
	\centering
	\begin{subfigure}{0.48\linewidth}
		\centering
		\includegraphics[width=0.95\linewidth]{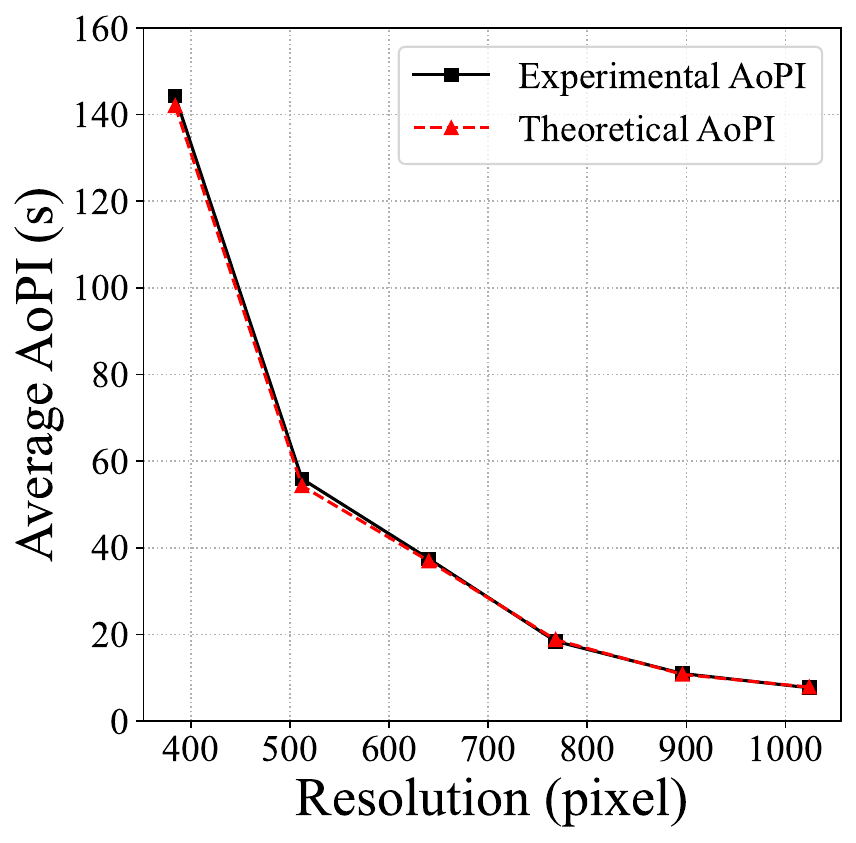}
        \caption{}
		\label{lcfsp_trade_c2}
	\end{subfigure}
	\caption{Comparison between the experimental and theoretical average AoPI on the GPU-based edge server under (a) FCFS and (b) LCFSP policy.}
	\label{et_age2}
\end{figure}

\subsection{Testbed Experiments}
In this section, we build a testbed to validate the effectiveness of the analytical results in Section \ref{ave_aoi} and further evaluate the performance of the proposed LBCD method.

\subsubsection{Analytical Results Validation}
\label{pv}
Our testbed is shown in Fig. \ref{pro}. We employ five Raspberry pi 3b+ as the cameras, which upload video frames to the edge servers via the campus WiFi. The edge servers are built upon two popular commercial hardware, i.e., the NVIDIA Jetson AGX Xavier board (act as GPU-based edge server) and Dell PowerEdge R740 (act as CPU-based edge server). On the Raspberry pi, the video frames are transmitted via HTTP post requests. On the edge server, a tornado server\footnote{https://www.tornadoweb.org/} is started to receive the uploaded frames. The YOLOv5m model \cite{yolov5m} is utilized for object recognition. Only frames with mean average precision (mAP) higher than 0.6 are considered as accurately recognized.

The theoretical and experimental AoPI under FCFS and LCFSP policy are shown in Fig. \ref{et_age} and \ref{et_age2}. Generally, the theoretical results deviate from the experimental ones by 3.33\% on average in Fig. \ref{et_age} and \ref{et_age2}. Besides, the theoretical optimal video configuration also has the best experimental AoPI. Hence, although frames' transmission/computation delay may not follow exponential distributions exactly in practical systems, our analysis results are still sufficiently accurate and effective to help adapt the video configurations.

We also note that the results in Fig. \ref{et_age2}(a) and \ref{et_age2}(b) are pretty similar. The main reason is that the transmission rate of video frames is much lower than the computation rate on the GPU-based edge server. Thus only very few frames can be preempted under the LCFSP policy. On the other hand, in Fig. \ref{et_age}, the computation rate of high-resolution frames (e.g., 1024p) is much lower on the CPU-based edge server. Hence, many frames must wait in a queue under the FCFS policy, leading to high AoPI.

\begin{figure}[tbp]
	\centering
	\begin{subfigure}{0.48\linewidth}
		\centering
		\includegraphics[width=0.95\linewidth]{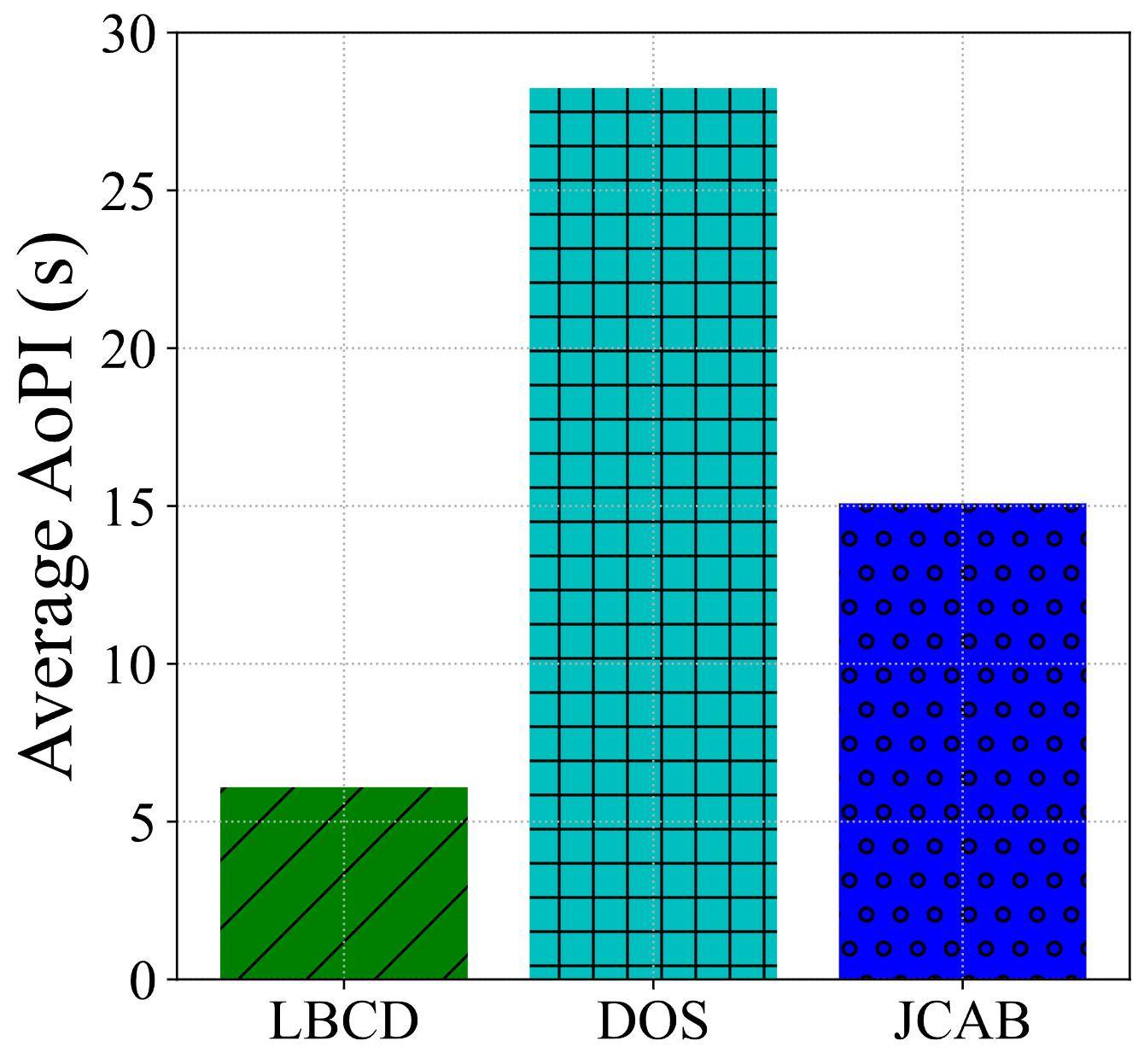}
        \caption{}
		\label{testbed_a}
	\end{subfigure}
	\centering
	\begin{subfigure}{0.48\linewidth}
		\centering
		\includegraphics[width=0.95\linewidth]{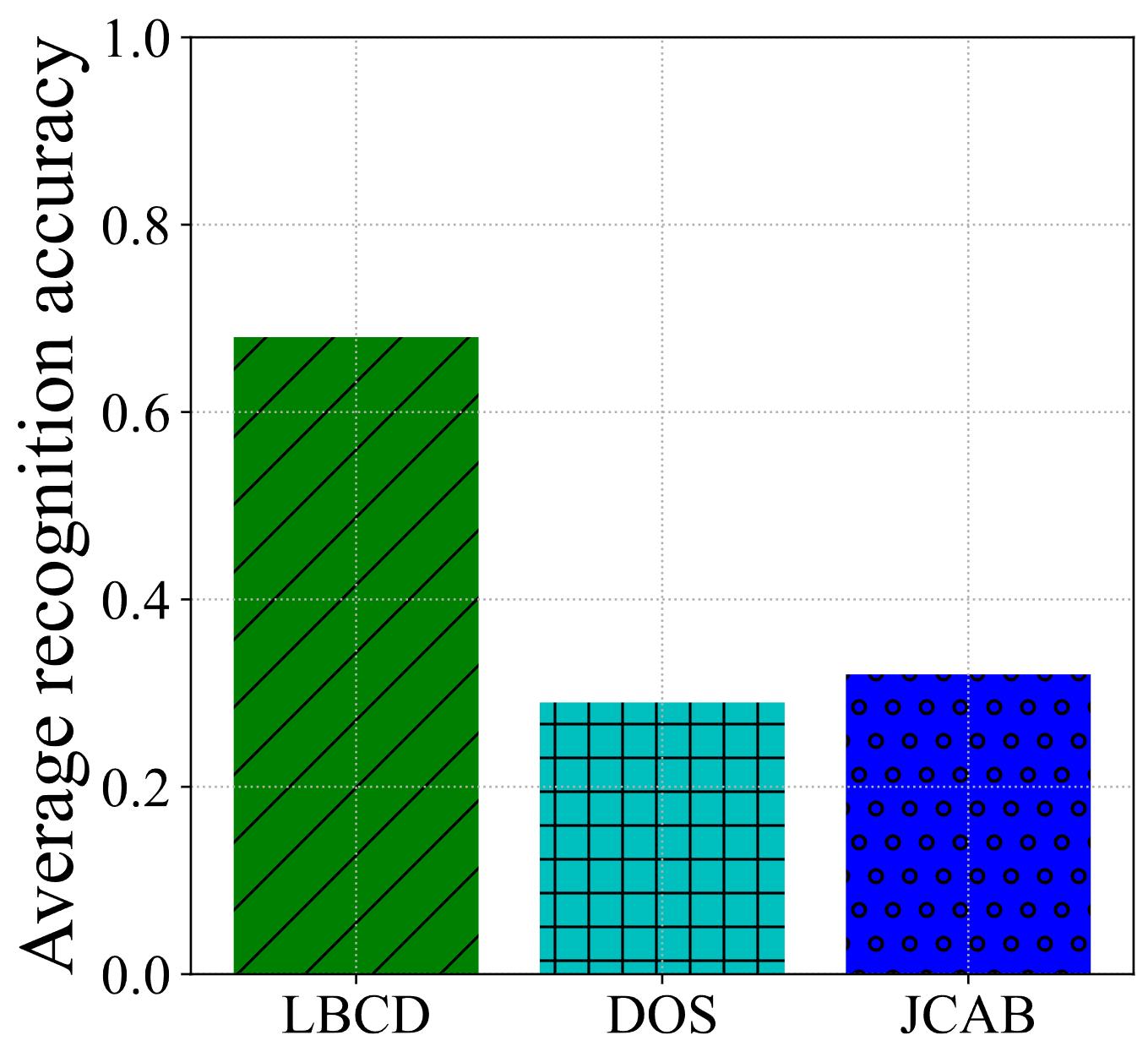}
        \caption{}
		\label{testbed_b}
	\end{subfigure}
	\caption{Comparison of (a) average AoPI and (b) average recognition accuracy on the testbed.}
	\label{testbed}
\end{figure}

\subsubsection{Testbed Evaluation}
\textcolor{black}{
The testbed consists of five cameras and two edge servers in this section. The first three cameras account for the object recognition task, while the remaining two cameras account for the segmentation task.} In the WiFi environment, we implement the bandwidth allocation strategies by controlling the uploading data rate between cameras and the edge servers using WonderShaper\footnote{https://github.com/magnific0/wondershaper}. Besides, we note that there currently lacks fine-grained GPU resource management tools (i.e., one task is equipped with one or more GPUs, rather than 10\% of a GPU). Hence, we conduct experiments only on the CPU-based edge servers (Dell PowerEdge R740). Specifically, we implement the computation resource allocation strategies by controlling the number of CPU cores used by different processes in PyTorch\footnote{https://pytorch.org/}. \textcolor{black}{
The average AoPI and recognition accuracy of different methods are shown in Fig. \ref{testbed}. As can be seen, LBCD achieves the best AoPI and meets the long-term recognition accuracy requirement (i.e., higher than 0.7). Specifically, LBCD reduces the average AoPI by 4.63X and 2.47X compared with the DOS and JCAB methods, respectively, demonstrating its superiority in practical systems.} 

In conclusion, according to the above simulations and testbed experiments, the proposed LBCD method can significantly outperform the state-of-the-art approaches in terms of AoPI and maintain the long-term recognition accuracy under various network conditions.

\section{Conclusion and Future Work}
\label{con_fut}
\textcolor{black}{
This paper investigated the joint edge server selection, video configuration adaptation, and resource allocation problem. Instead of struggling to balance the heterogeneous delay and accuracy metrics, we introduced a timeliness concept, AoPI, to depict the integrated impact of transmission/computation efficiency and recognition accuracy. The closed-form expressions of AoPI are derived under FCFS and LCFSP computation policies. Based on the analytical results, an optimization problem is formulated to minimize the long-term average AoPI across all cameras. An online LBCD method is proposed to handle the problem requiring no future knowledge. Simulations and testbed experimental results show that the proposed method reduces AoPI by up to 10.94X compared with the state-of-the-art baselines. In future work, we will extend the AoPI analysis under general delay distributions and consider the energy consumption requirements of practical systems.}



\end{document}